\documentclass{book}
\usepackage{makeidx}
\usepackage{elsevierg}
\usepackage{graphics}
\usepackage{graphicx}
\usepackage{amssymb}
\usepackage{latexsym}
\newcommand{\uline}{\vrule height.06ex depth.02ex width.6em}
\newcommand{\mvee}{\vee\kern-.69em\uline}

\makeindex
\collection
\begin{document}

\paper{Quantum Logic and Quantum Computation}
{Mladen Pavi\v ci\'c and Norman D.~Megill}

\section{Introduction}
\label{sec:intro}

In the literature on quantum computation, quantum logic
means an algebra of the qubits and quantum gates of a quantum
computer.~\cite{berman-ssqc-book-05-ql,bouw-ek-zeil-book-00,%
fio-prl04,franson-02,nielsen-chuang-ql482,pavicic-book-05,%
shapiro-prl,zurek-96} This quantum logic of qubits
(also called {\em quantum computational
logic}~\cite{chiara-qcomp04,gudder-03}) is a formalism of
finite tensor products of two-dimensional Hilbert spaces and
will be the subject matter of Volume 3 of this Handbook.

Quantum logic of qubits will not be considered here, because
this volume of the Handbook deals with quantum logics defined
as algebras related to a complete description of quantum
systems, from orthomodular posets to Hilbert lattices.
A complete description of a quantum system, say a molecule,
includes not only spins---as with qubits---but also
positions, momenta, and potentials of nucleons and
electrons, and this, in the standard approach, requires
infinite-dimensional Hilbert spaces. In the second half
of the 20th century, numerous attempts to reduce the
latter Hilbert space formalism to various types of
algebras have been put forward.~\cite{holl95} The main
idea behind these attempts was to relate Hilbert
space observables directly to experimental setups and
results.~\cite{jauch,ludwig-book-1,ludwig-book-2,piron-book}

The latter idea has not come true, but mathematically
the project has been a success. In particular, the Hilbert
lattice has been proved isomorphic to the set of
subspaces of an infinite-dimensional Hilbert space.
So, in an attempt to treat general quantum systems
with the help of a quantum computer, we might venture to
introduce such an algebraic description of the systems
directly into it. However, as with classical problems, we
have to translate a description of quantum systems into a
language a quantum computer would understand. To make this
point, before we dwell on quantum systems, we shall briefly
review how we can make such a translation for a classical
problem to be computed on a quantum computer.

One of the most successful quantum computing algorithms
so far is Shor's algorithm for the classical problem
of factoring numbers.~\cite{shor}
Factoring numbers with classical algorithms on classical
computers is conjectured to be a problem of exponential
complexity with respect to the number of bits.
To verify (by the brute force approach) whether $x,y>1$
exist such that $xy=N$, we have to check all possible $x$'s
starting with $x=2$ and ending (in the most
unfavourable case) with $x=\sqrt{N}$.
The number of checks obviously does not rise exponentially
with $N$. When we say that the time needed
to carry out the checking rises exponentially, we mean with
respect to the number of bits $n$ required to handle the 
divisions within a digital computer, where 
$N\approx 2^n$.~\cite{pavicic-book-05}

The way of handling operations in a computer
is what differentiates a classical from a quantum computer
and what enables the latter one to speed up exponentially 
given computations. To see
this let us look at the {\it all-optical} ``calculator''
shown in Fig.~\ref{fig:mz} proposed by Johann
Summhammer.~\cite{summhammer}

\begin{figure}[hbt]
\begin{center}
\includegraphics[width=0.6\textwidth]{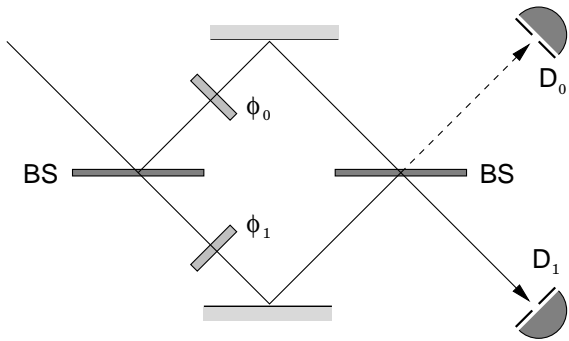}
\end{center}
\caption{{\em Physical calculation} by means of a Mach-Zehnder
interferometer. When the phase difference $\phi_0-\phi_1$ is
equal to an integer multiplier of $\pi$,
the photon ideally will never end up in detector
D$_0$. {\rm BS} are 50:50 beam splitters.}
\label{fig:mz}
\end{figure}

Photons enter the Mach-Zehnder interferometer one by one
and interfere de\-pending on the phase difference $\phi_0-\phi_1$
imposed by the phase shifters shown in the figure.
The probability of detector D$_1$ registering a photon is
\begin{eqnarray}
p=\cos^2{{\phi_1-\phi_0}\over 2}.
\label{eq:prob-mz-cos}
\end{eqnarray}
Hence, when the phase difference $\phi_0-\phi_1$ is
equal to an integer multiplier of $\pi$, we have $p=1$,
i.e., a constructive interference at the lower exit and
a destructive one at the upper one. In other words, ideally,
detector D$_1$ will always click and detector D$_0$ will
never click.

To factor a number $N$, let us increase the phase shift in
discrete steps $2\pi/k$ so as to have $\phi_0-\phi_1=2\pi N/k$.
Until we arrive at an integer $k$ which factors $N$, we will always
have some clicks in D$_0$. Our qubit---photon state---will
be in a superposition that is not completely destructive
with respect to the upper exit. By repeating each step many
times, say $N$ times, we will altogether need---in order to
reach and confirm a completely destructive interference with
respect to the upper exit---to carry out $N\sqrt{N}$ checks,
to make sure that the response of D$_0$ is
negligible when we reach a $k$ for which $N/k$ is an integer and
$p$ given by Eq.~(\ref{eq:prob-mz-cos}) equals 1.

We see that the way of introducing and dividing numbers into
this {\em physical computer} is essentially different from the
one we use with digital computers. In a classical computer,
the bigger the numbers are, the more bits and therefore more
transistors we have to employ for handling their factorization.
The number of transistors and gates required to handle each
division grows polynomially with the number of bits $n$.

In a Mach-Zehnder interferometer, we carry out each division by
picking a particular phase difference. So, it is always just one
phase difference, irrespective of the size of the number.\footnote{Of
course, up to a realistic limit of the interferometer---it cannot
discern phases that correspond to numbers bigger than $10^{10}$,
but it can be integrated together with other linear optics
elements into an {\em all-optical} quantum computer.}
If we take the photon within the interferometer to be a
{\em quantum bit}, in quantum computation parlance called
{\em qubit}, and the Mach-Zehnder interferometer
to be a {\em quantum logic gate}, or simply a
{\em quantum gate}, we obtain a single gate acting as
a single computing unit, providing us with the desired result
within a single step. We denote the state
of a photon exiting from the lower sides of the Mach-Zehnder
beam splitters by $|1\rangle$ and one exiting from their
upper sides by $|0\rangle$. The qubit (photon) can be
in any of infinitely many superpositions
$\alpha|0\rangle+\beta|1\rangle$ within the interferometer,
one of which provides us with a definite result with the
probability equal to 1.

Now, what enables an exponential speedup of the quantum factoring
tasks is an al\-go\-rithm---Shor's algorithm---that makes use
of a superposition of qubit states so as to reduce the problem of
searching for factors to searching for a period of the wave function
representing the superposition. When applied, it reduces
the required number of checks to one polynomial in
$\log(N)$.~\cite{pittenger-book} This algorithm as well as
all other known quantum algorithms are based on Fourier transforms,
and the main additional feature of quantum gates is that
they can perform Fourier transforms. The next feature we
require is scalability, so that adding new gates preserves
the achieved speedup while increasing the computational
power. Linear optics elements---the Mach-Zehnder interferometer
being one of them---can be integrated into a scalable all-optical
quantum computer.~\cite{browne05,ralph05}
Similar scaling up can be achieved with ion, quantum-dot,
QED, and Kane quantum computers by using recent
hardware and software blueprints, at least
in principle.~\cite{pavicic-book-05}
At present, quantum computation relies on the way we introduce and
encode the input data into states of qubits within a quantum
computer as well as on the algorithms we apply to the states.
There are still only a few such algorithms, but for
various classical problems we shall most probably arrive at
new applications gradually, as was the case with classical
computation (it took half a century to reach a digital
implementation of 3D animation and voice recognition, for
example).

There is, however, an application that seems to
be radically different from all the others, and this is
the quantum computation and simu\-lation of quantum systems.
Classically, we compute the properties of an atom or molecule
by solving a Schr\"o\-din\-ger equa\-tion with the help of
sophisticated algorithms for approximating and solving the
equation, all of which are of exponential complexity with
respect to the number of observables. For a quantum
computer, several algorithms for solving a Schr\"odinger
equation that provide an exponential speed increase
with respect to classical computers have been
proposed.\ \cite{abrams-lloyd-99,schr-simul,gramss,zalka-98}
They start with rather simple wave functions,
discretize them, and then introduce them into the
Schr\"odinger equation, which then reduces to an
eigenvalue problem that can be solved in a Fourier
transform approach analogous to the one we used to
factor numbers. In the case of more general Schr\"odinger
equations, though, we no longer have an obvious and
straightforward algorithm---no algorithms are known that
implement Fourier transforms for simulating and determining
the evolution of general arbitrary quantum systems.

However, if we found a quantum algebra for describing
quantum systems, such as atoms and molecules, which a quantum
computer could ``read'' directly, then it would instantly
simulate the systems, tremendously speeding up its ``calculation,''
i.e., obtaining information on its behaviour. No special
algorithm would be needed. We could
think of simulating existing and still non-existing molecules
under chosen conditions. How can we achieve such a simulation?

When we talk about quantum systems and its theoretical
Hilbert space, we know that there is
a Hilbert lattice that is isomorphic to the set of
subspaces of a particular infinite-dimensional Hilbert
space and that we can establish a correspondence
between elements of the lattice and solutions of
a Schr\"odinger equation that corresponds to such a
Hilbert space. But there is an essential problem here.
Any Hilbert lattice is a structure based on
first-order predicate calculus,
and we simply cannot have a constructive pro\-ced\-ure
to introduce state\-ments like {\em there is} or {\em for all}
into a computer. Unlike with the Mach-Zehnder computer above,
we do not have a recipe for introducing states of an
arbitrary quantum system into a quantum computer.

What we might do, instead, is find classes of polynomial
lattice equations that can serve in place of quantified
statements. And in this chapter we
are going to review how far we have advanced down this
road, following \cite{mpoa99} and \cite{pavicic-book-05}.
If we can eventually establish a correspondence
between such equations and solutions of the
Schr\"odinger equation, i.e., general wave functions,
then we should be able to reduce any quantum problem
to a polynomially complex eigenvalue problem.
Whether the project can be carried out successfully
awaits future developments, but this is the case with
all projects in quantum computing.

In the next section, we deal with quantum logic defined
as a Hilbert lattice. Quantum logic so defined is only
one of the possible models of quantum logics con\-sider\-ed
in our chapter in Volume~1.~\cite{pm-ql-l-hql1} In Section
\ref{sec:greechie} we present the Greechie diagrams,
in Section  \ref{sec:goe} generalized orthoarguesian equations
and in Sections  \ref{sec:ge}, \ref{sec:mge}, and  \ref{sec:mayet}
Godowski, Mayet-Godowski, and Mayet's E-equations, respectively.
We end the chapter with a conclusion and open problems.

\newpage

\section{Hilbert Lattice}
\label{sec:hl}

A Hilbert lattice is a special kind of an orthomodular
lattice, OML, which we introduced and defined in Definition 2.6
in our chapter in Volume 1.~\cite{pm-ql-l-hql1}
The axioms added to an OML to make it represent Hilbert
space are (as one example of several slightly
different axiomatizations) the following
ones.~\cite{beltr-cass-book,kalmb86}

\begin{definition}\label{def:hl}\footnote{For additional
definitions of the terms used in this section see
Refs.~\cite{beltr-cass-book,holl95,kalmb86}}
An orthomodular lattice which satisfies the following
con\-di\-tions is a {\em Hilbert lattice}, $\mathcal{HL}$.
\begin{enumerate}
\item {\em Completeness:\/}
The meet and join of any subset of
an $\mathcal{HL}$ exist.
\item {\em Atomicity:\/}
Every non-zero element in an $\mathcal HL$ is greater
than or equal to an atom. (An atom $a$ is a non-zero lattice element
with $0< b\le a$ only if $b=a$.)
\item {\em Superposition principle:\/}
(The atom $c$
is a superposition of the atoms $a$ and $b$ if
$c\ne a$, $c\ne b$, and $c\le a\cup b$.)
\begin{description}
\item[{\rm (a)}] Given two different atoms $a$ and $b$, there is at least
one other atom $c$, $c\ne a$ and $c\ne b$, that is a superposition
of $a$ and $b$.
\item[{\rm (b)}] If the atom $c$ is a superposition of  distinct atoms
$a$ and $b$, then atom $a$ is a superposition of atoms $b$ and $c$.
\end{description}
\item {\em Minimal length:\/} The lattice contains at least
three elements $a,b,c$ satisfying: $0<a<b<c<1$.
\end{enumerate}
\end{definition}

These conditions imply an infinite number of atoms in $\mathcal HL$ as
shown by Ivert and Sj{\"o}din.\ \cite{ivertsj}

One can prove the following theorem
\cite{mackey,maclaren,varad}.

\begin{theorem}\label{th:repr}For every Hilbert lattice
$\mathcal{HL}$ there exists a field $\mathcal K$ and a Hilbert space
$\mathcal H$ over $\mathcal K$ such that the set of closed
subspaces of the Hilbert space, ${\mathcal C}({\mathcal H})$ is
ortho-isomorphic
\index{ortho-isomorphism}%
to $\mathcal HL$.

Conversely, let $\mathcal H$ be an infinite-dimensional Hilbert space
over a field $\mathcal K$ and let
\begin{eqnarray}
{\mathcal C}({\mathcal H})\ {\buildrel\rm def\over =}\ \{ {\mathcal X}\
\subseteq {\mathcal H}\ | \>{\mathcal X}^{\perp\perp}={\mathcal X}\}
\end{eqnarray}
be the set of all closed subspaces of $\mathcal H$.
Then ${\mathcal C}({\mathcal H})$ is a Hilbert lattice
relative to:
\begin{eqnarray}
a\cap b\ =\ {\mathcal X}_a\cap {\mathcal X}_b
\qquad\quad {\rm and}\qquad\quad a\cup b\
 =\ ({\mathcal X}_a+{\mathcal X}_b)^{\perp\perp}.\qquad
\end{eqnarray}
\end{theorem}

In order to determine the field over which the Hilbert space
in Theorem \ref{th:repr} is defined, we make use of the
following theorem proved by Maria Pia Sol{\`e}r
\cite{soler,holl95}.

\begin{theorem}\label{th:sol}
The Hilbert space $\mathcal H$ from Theorem \ref{th:repr} is
an infinite-dimen\-sional Hilbert space defined over
a real, complex, or quaternion (skew) field
if the following condition is met:
\begin{itemize}
\item {\em Infinite orthonormality:} $\mathcal HL$ contains a
countably infinite sequence of orthonormal elements.
\end{itemize}
\end{theorem}

Thus we do arrive at a full Hilbert space, but the axioms
for the Hilbert lattices that we used for this purpose are
rather involved. This is because in the past, the axioms were
simply read off from the Hilbert space structure and were formulated
as quantified statements of the first order that cannot
be implemented into a quantum computer.

\section{Greechie Diagrams}\label{sec:greechie}

The Hilbert lattice equations that we will be describing in subsequent
sections will require some method for proving that they are independent
from the equations for OMLs.  This will show that these equations indeed
extend the equational theory for Hilbert lattices beyond that provided
by just the OML equations.  We will usually show the independence by
exhibiting finite OMLs in which the new Hilbert lattice equations fail.
Typically, these counterexample OMLs are very large lattices with dozens
of nodes, and it is inconvenient to represent them with standard lattice
(Hasse) diagrams.  Instead, we will use a much more compact method for
representing OMLs called {\em Greechie diagrams}.  Because of their
importance as a tool, we will describe them in some detail in this
section.

The following definitions and theorem we take over from Kalmbach
\cite{kalmb83} and Svozil and Tkadlec \cite{svozil-tkadlec}.
Definitions in the framework of \it quantum logics\/ \rm
($\sigma$-orthomodular posets) the reader can find in the book of
Pt\'ak and Pulmannov{\'a}. \cite{ptak-pulm}

\begin{definition}\label{D:diagram}
A\/ {\em diagram} is a pair $(V,E)$, where $V\ne\emptyset$ is a set of\/
{\em atoms} (drawn as points) and
$E\subseteq {\rm exp}\,V\>\backslash\,\{\emptyset\}$ is a set of\/
{\em blocks} (drawn as line segments connecting corresponding points).
A\/ {\em loop} of order $n\ge 2$ ($n$ being a natural number) in a
diagram $(V,E)$ is a sequence $(e_1,\dots e_b)\in E^n$ of mutually
different blocks such that there are mutually distinct
atoms $\nu_1,\dots,\nu_n$ with $\nu_i\in e_i\cap e_{i+1}\
(i=1,\dots,n,\ e_{n+1}=e_1)$.
\end{definition}

\begin{definition}\label{D:greechie-diagram}
A\/ {\em Greechie diagram} is a diagram satisfying the following
con\-di\-tions:
\begin{enumerate}
\item[(1)] Every atom belongs to at least one block.
\item[(2)] If there are at least two atoms then every block is at least
2-element.
\item[(3)] Every block which intersects with another block is at least
3-element.
\item[(4)] Every pair of different blocks intersects in at most one atom.
\item[(5)] There is no loop of order 3.
\end{enumerate}
\end{definition}

\begin{theorem}\label{th:loop-lemma}
For every Greechie diagram with only finite blocks there is exactly
one (up to an isomorphism) orthomodular poset such that there are
one-to-one correspondences between atoms and atoms and between
blocks and blocks that preserve incidence relations.
The poset is a lattice if and only if the Greechie diagram has
no loops of order~4.
\end{theorem}

In general, Greechie diagrams correspond to Boolean algebras ``pasted''
together.  First we will show examples of individual blocks in
order to illustrate how they correspond to Boolean algebras.  Then we
will introduce the concepts needed to understand how Boolean
algebras can be interconnected to represent more general OMLs.

The Hasse and Greechie diagrams for the Boolean algebras
corresponding to 2-, 3-, and
4-atom blocks are shown in Fig.~\ref{fig:greechie-1}.  The Greechie
diagram for a given lattice may be drawn in several equivalent ways:
Fig.~\ref{fig:greechie-2}\ shows the same Greechie diagram drawn in two
different ways, along with the corresponding Hasse diagram.  {}From the
definitions we see that the ordering of the atoms on a block does not
matter, and we may also draw blocks using arcs as well as straight lines
as long as the blocks remain clearly distinguishable.

\begin{figure}[htbp]\centering
  \setlength{\unitlength}{1pt}
  \begin{picture}(330,190)(0,0)
    \put(0,0){
      \begin{picture}(40,10)(0,0)
        \put(0,0){\line(1,0){40}}
        \put(0,0){\circle*{5}}
        \put(40,0){\circle*{5}}
        \put(0,10){\makebox(0,0)[b]{$x$}}
        \put(40,10){\makebox(0,0)[b]{$x'$}}
      \end{picture}
    }

    \put(0,40){
      \begin{picture}(40,60)(0,0)
        \put(20,0){\line(-2,3){20}}
        \put(20,0){\line(2,3){20}}
        \put(20,60){\line(-2,-3){20}}
        \put(20,60){\line(2,-3){20}}

        \put(20,-5){\makebox(0,0)[t]{$0$}}
        \put(-5,30){\makebox(0,0)[r]{$x$}}
        \put(45,30){\makebox(0,0)[l]{$x'$}}
        \put(20,65){\makebox(0,0)[b]{$1$}}

        \put(20,0){\circle*{3}}
        \put(0,30){\circle*{3}}
        \put(40,30){\circle*{3}}
        \put(20,60){\circle*{3}}
      \end{picture}
    } 
    \put(90,0){
      \begin{picture}(60,10)(0,0)
        \put(0,0){\line(1,0){60}}
        \put(0,0){\circle*{5}}
        \put(30,0){\circle*{5}}
        \put(60,0){\circle*{5}}
        \put(0,10){\makebox(0,0)[b]{$x$}}
        \put(30,10){\makebox(0,0)[b]{$y$}}
        \put(60,10){\makebox(0,0)[b]{$z$}}
      \end{picture}
    }
    \put(90,40){
      \begin{picture}(60,90)(0,0)

        \put(30,0){\line(-1,1){30}}
        \put(30,0){\line(0,1){30}}
        \put(30,0){\line(1,1){30}}
        \put(30,90){\line(-1,-1){30}}
        \put(30,90){\line(0,-1){30}}
        \put(30,90){\line(1,-1){30}}
        \put(0,30){\line(1,1){30}}
        \put(30,30){\line(1,1){30}}
        \put(0,30){\line(2,1){60}}
        \put(30,30){\line(-1,1){30}}
        \put(60,30){\line(-1,1){30}}
        \put(60,30){\line(-2,1){60}}

        \put(30,-5){\makebox(0,0)[t]{$0$}}
        \put(-5,30){\makebox(0,0)[r]{$x$}}
        \put(38,30){\makebox(0,0)[l]{$y$}}
        \put(65,30){\makebox(0,0)[l]{$z$}}
        \put(-5,60){\makebox(0,0)[r]{$x'$}}
        \put(38,60){\makebox(0,0)[l]{$y'$}}
        \put(65,60){\makebox(0,0)[l]{$z'$}}
        \put(30,95){\makebox(0,0)[b]{$1$}}

        \put(0,30){\circle*{3}}
        \put(30,30){\circle*{3}}
        \put(60,30){\circle*{3}}
        \put(0,60){\circle*{3}}
        \put(30,60){\circle*{3}}
        \put(60,60){\circle*{3}}
        \put(30,0){\circle*{3}}
        \put(30,90){\circle*{3}}
      \end{picture}
    } 

    \put(200,0){
      \begin{picture}(100,10)(0,0)
        \put(20,0){\line(1,0){80}}
        \put(20,0){\circle*{5}}
        \put(46.67,0){\circle*{5}}
        \put(73.33,0){\circle*{5}}
        \put(100,0){\circle*{5}}
        \put(20,10){\makebox(0,0)[b]{$w$}}
        \put(46.67,10){\makebox(0,0)[b]{$x$}}
        \put(73.33,10){\makebox(0,0)[b]{$y$}}
        \put(100,10){\makebox(0,0)[b]{$z$}}
      \end{picture}
    }

    \put(200,40){
      \begin{picture}(100,140)(0,0)
        \put(60,0){\line(-2,3){20}}
        \put(60,0){\line(2,3){20}}
        \put(60,0){\line(-4,3){40}}
        \put(60,0){\line(4,3){40}}
        \put(60,140){\line(-2,-3){20}}
        \put(60,140){\line(2,-3){20}}
        \put(60,140){\line(-4,-3){40}}
        \put(60,140){\line(4,-3){40}}

        \put(20,30){\line(-1,2){20}}
        \put(20,30){\line(0,1){40}}
        \put(20,30){\line(1,2){20}}
        \put(40,30){\line(-1,1){40}}
        \put(40,30){\line(1,1){40}}
        \put(40,30){\line(3,2){60}}
        \put(80,30){\line(-3,2){60}}
        \put(80,30){\line(0,1){40}}
        \put(80,30){\line(1,1){40}}
        \put(100,30){\line(-3,2){60}}
        \put(100,30){\line(0,1){40}}
        \put(100,30){\line(1,2){20}}

        \put(100,110){\line(1,-2){20}}
        \put(100,110){\line(0,-1){40}}
        \put(100,110){\line(-1,-2){20}}
        \put(80,110){\line(1,-1){40}}
        \put(80,110){\line(-1,-1){40}}
        \put(80,110){\line(-3,-2){60}}
        \put(40,110){\line(3,-2){60}}
        \put(40,110){\line(0,-1){40}}
        \put(40,110){\line(-1,-1){40}}
        \put(20,110){\line(3,-2){60}}
        \put(20,110){\line(0,-1){40}}
        \put(20,110){\line(-1,-2){20}}

        \put(60,140){\circle*{3}}
        \put(20,110){\circle*{3}}
        \put(40,110){\circle*{3}}
        \put(80,110){\circle*{3}}
        \put(100,110){\circle*{3}}
        \put(0,70){\circle*{3}}
        \put(20,70){\circle*{3}}
        \put(40,70){\circle*{3}}
        \put(80,70){\circle*{3}}
        \put(100,70){\circle*{3}}
        \put(120,70){\circle*{3}}
        \put(20,30){\circle*{3}}
        \put(40,30){\circle*{3}}
        \put(80,30){\circle*{3}}
        \put(100,30){\circle*{3}}
        \put(60,0){\circle*{3}}

        \put(60,-5){\makebox(0,0)[t]{$0$}}
        \put(105,28){\makebox(0,0)[l]{$z$}}
        \put(75,28){\makebox(0,0)[r]{$y$}}
        \put(48,28){\makebox(0,0)[l]{$x$}}
        \put(15,28){\makebox(0,0)[r]{$w$}}
        \put(105,110){\makebox(0,0)[l]{$w'$}}
        \put(75,110){\makebox(0,0)[r]{$x'$}}
        \put(48,113){\makebox(0,0)[l]{$y'$}}
        \put(15,110){\makebox(0,0)[r]{$z'$}}
        \put(60,145){\makebox(0,0)[b]{$1$}}


      \end{picture}
    } 

  \end{picture}
  \caption{Greechie diagrams for Boolean lattices $2^2$, $2^3$, and $2^4$,
   labelled with the atoms of their corresponding Hasse diagrams shown
   above them.  ($2^4$ was adapted from \cite[Fig.~18, p.~84]{beran}.)
\label{fig:greechie-1}}
\end{figure}
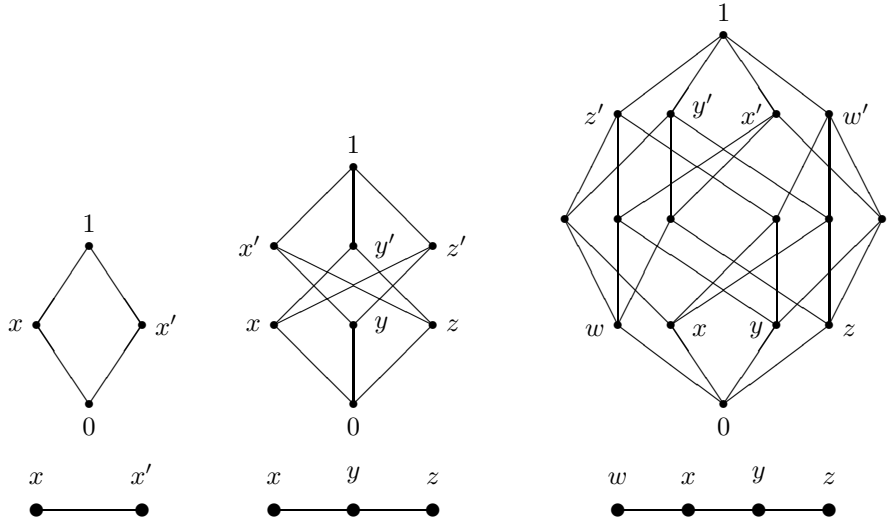

Recall that a {\em poset\/} (partially ordered set) is a set with an
associated ordering relation that is reflexive ($a\le a$), antisymmetric
($a\le b,b\le a$ imply $a=b$), and transitive ($a\le b,b\le c$ imply
$a\le c$).  An {\em orthoposet\/} is a poset with lower and upper bounds
$0$ and $1$ and an operation~$'$ satisfying (i) if $a\le b$ then $b'\le
a'$; (ii) $a''=a$; and (iii) the infimum $a\cap a'$ and the supremum
$a\cup a'$ exist and are $0$ and $1$ respectively.  A {\em lattice\/} is
a poset in which any two elements have an infimum and a supremum.  An
orthoposet is {\em orthomodular\/} if $a\le b$ implies (i) the supremum
$a\cup b'$ exists and (ii) $a\cup(a'\cap b)=b$.  A lattice is {\em
orthomodular\/} if it is also an orthomodular poset.  For example,
Boolean algebras such as those of Fig.~\ref{fig:greechie-1} are
orthomodular lattices.  A {\em $\sigma$-orthomodular poset} is an
orthomodular poset in which every countable subset of elements has
a supremum.  An {\em atom\/} of an orthoposet is an element
$a\ne 0$ such that $b<a$ implies $b=0$.

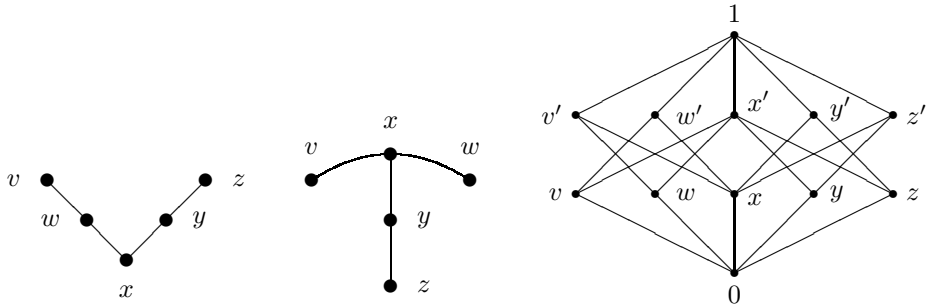
\begin{figure}[htbp]\centering
  \setlength{\unitlength}{1pt}
  \begin{picture}(330,100)(0,0)

    \put(0,10){
      \begin{picture}(60,80)(0,0)
        \put(0,30){\line(1,-1){30}}
        \put(30,0){\line(1,1){30}}
        \put(0,30){\circle*{5}}
        \put(15,15){\circle*{5}}
        \put(30,0){\circle*{5}}
        \put(45,15){\circle*{5}}
        \put(60,30){\circle*{5}}
        \put(-10,30){\makebox(0,0)[r]{$v$}}
        \put(5,15){\makebox(0,0)[r]{$w$}}
        \put(30,-10){\makebox(0,0)[t]{$x$}}
        \put(55,15){\makebox(0,0)[l]{$y$}}
        \put(70,30){\makebox(0,0)[l]{$z$}}
      \end{picture}
    } 

    \put(100,0){
      \begin{picture}(60,80)(0,0)
        \qbezier(0,40)(30,60)(60,40)
        \put(30,0){\line(0,1){50}}
        \put(30,0){\circle*{5}}
        \put(30,25){\circle*{5}}
        \put(30,50){\circle*{5}}
        \put(0,40){\circle*{5}}
        \put(60,40){\circle*{5}}
        \put(0,50){\makebox(0,0)[b]{$v$}}
        \put(30,60){\makebox(0,0)[b]{$x$}}
        \put(60,50){\makebox(0,0)[b]{$w$}}
        \put(40,25){\makebox(0,0)[l]{$y$}}
        \put(40,0){\makebox(0,0)[l]{$z$}}
      \end{picture}
    } 

    \put(200,5){
      \begin{picture}(60,90)(0,0)
        \put(60,90){\line(-2,-1){60}}
        \put(60,90){\line(-1,-1){30}}
        \put(60,90){\line(0,-1){30}}
        \put(60,90){\line(1,-1){30}}
        \put(60,90){\line(2,-1){60}}
        \put(60,0){\line(-2,1){60}}
        \put(60,0){\line(-1,1){30}}
        \put(60,0){\line(0,1){30}}
        \put(60,0){\line(1,1){30}}
        \put(60,0){\line(2,1){60}}
        \put(0,60){\line(1,-1){30}}
        \put(0,60){\line(2,-1){60}}
        \put(30,60){\line(-1,-1){30}}
        \put(30,60){\line(1,-1){30}}
        \put(60,60){\line(-2,-1){60}}
        \put(60,60){\line(-1,-1){30}}
        \put(60,60){\line(1,-1){30}}
        \put(60,60){\line(2,-1){60}}
        \put(90,60){\line(-1,-1){30}}
        \put(90,60){\line(1,-1){30}}
        \put(120,60){\line(-2,-1){60}}
        \put(120,60){\line(-1,-1){30}}

        \put(60,-5){\makebox(0,0)[t]{$0$}}
        \put(-5,30){\makebox(0,0)[r]{$v$}}
        \put(38,30){\makebox(0,0)[l]{$w$}}
        \put(65,28){\makebox(0,0)[l]{$x$}}
        \put(96,30){\makebox(0,0)[l]{$y$}}
        \put(125,30){\makebox(0,0)[l]{$z$}}
        \put(-5,60){\makebox(0,0)[r]{$v'$}}
        \put(38,60){\makebox(0,0)[l]{$w'$}}
        \put(65,65){\makebox(0,0)[l]{$x'$}}
        \put(96,62){\makebox(0,0)[l]{$y'$}}
        \put(125,60){\makebox(0,0)[l]{$z'$}}
        \put(60,95){\makebox(0,0)[b]{$1$}}

        \put(60,0){\circle*{3}}
        \put(0,30){\circle*{3}}
        \put(30,30){\circle*{3}}
        \put(60,30){\circle*{3}}
        \put(90,30){\circle*{3}}
        \put(120,30){\circle*{3}}
        \put(0,60){\circle*{3}}
        \put(30,60){\circle*{3}}
        \put(60,60){\circle*{3}}
        \put(90,60){\circle*{3}}
        \put(120,60){\circle*{3}}
        \put(60,90){\circle*{3}}
      \end{picture}
    } 

  \end{picture}
  \caption{Two different ways of drawing the same Greechie diagram, and its
    corresponding Hasse diagram.
\label{fig:greechie-2}}
\end{figure}

In the literature, there are several different definitions of a Greechie
diagram.  For example, Beran (\cite[p.~144]{beran}) forbids 2-atom
blocks.  Kalmbach (\cite[p.~42]{kalmb83}) as well as Pt\'ak and
Pulmannov{\'a} \cite[p.~32]{ptak-pulm} include all diagrams with 2-atom
blocks connected to other blocks as long as the resulting pasting
corresponds to an orthoposet.  However, the case of 2-atom blocks
connected to other blocks is somewhat complicated; for example, the
definition of a loop in Definition~\ref{D:diagram} must be modified
(e.g.~\cite[p.~42]{kalmb83}) and no longer corresponds to the simple
geometry of a drawing of the diagram.  The definition of a Greechie
diagram also becomes more complicated; for example a pentagon (or any
$n$-gon with an odd number of sides) made out of 2-atom blocks is not a
Greechie diagram (i.e.\ does not correspond to any orthoposet).

The definition of Svozil and Tkadlec \cite{svozil-tkadlec} that we
adopt, Definition~\ref{D:greechie-diagram}, excludes 2-atom blocks
connected to other blocks.  It turns out that all orthomodular posets
representable by Kalmbach's definition can be represented with the
diagrams allowed by Svozil and Tkadlec's definition.  But the latter
definition eliminates the special treatment of 2-atom blocks connected
to other blocks and in particular simplifies any computer program
designed to process Greechie diagrams.

Svozil and Tkadlec's definition further restricts Greechie diagrams to
those diagrams representing orthoposets that are orthomodular by
forbidding loops of order less than~4, unlike the definitions of Beran
and Kalmbach.  The advantage appears to be mainly for convenience, as we
obtain only those Greechie diagrams that correspond to what are
sometimes called ``quantum logics'' ($\sigma$-orthomodular posets).
(We note that the term ``quantum logic'' is also used to denote a
pro\-posi\-tional calculus based on orthomodular or weakly orthomodular
lattices. \cite{mpcommp99})

The definition allows for Greechie diagrams whose blocks are not
connected.  In Fig.~\ref{fig:greechie-3} we show the Greechie diagram
for the {\em Chinese lantern\/} MO2 using unconnected 2-atom blocks.
This example also illustrates that even when the blocks are unconnected,
the properties of the resulting orthoposet are not just a simple
combination of the properties of their components (as one might na\"\i
vely suppose), because we are adding disjoint sets of incomparable nodes
to the ortho\-poset.  As is well-known (\cite[p.~16]{kalmb83}), MO2 is not
distributive, unlike the Boolean blocks it is built from.

\begin{figure}[htbp]\centering
  \setlength{\unitlength}{1pt}
  \begin{picture}(240,80)(0,0)

    \put(0,30){
      \begin{picture}(40,10)(0,0)
        \put(0,0){\line(1,0){20}}
        \put(0,0){\circle*{5}}
        \put(20,0){\circle*{5}}
        \put(0,10){\makebox(0,0)[b]{$x'$}}
        \put(20,10){\makebox(0,0)[b]{$x$}}
        \multiput(20,0)(4,0){10}{\line(1,0){2}}
        \put(60,0){\line(1,0){20}}
        \put(60,0){\circle*{5}}
        \put(80,0){\circle*{5}}
        \put(60,10){\makebox(0,0)[b]{$y$}}
        \put(80,10){\makebox(0,0)[b]{$y'$}}

        \put(40,-20){\makebox(0,0)[c]{(a)}}
      \end{picture}
    }

    \put(170,0){

      \begin{picture}(90,80)(0,0)
        \put(40,0){\line(-2,3){20}}
        \put(40,0){\line(2,3){20}}
        \put(40,0){\line(-4,3){40}}
        \put(40,0){\line(4,3){40}}
        \put(40,60){\line(-2,-3){20}}
        \put(40,60){\line(2,-3){20}}
        \put(40,60){\line(-4,-3){40}}
        \put(40,60){\line(4,-3){40}}

        \put(40,-5){\makebox(0,0)[t]{$0$}}
        \put(25,30){\makebox(0,0)[l]{$x$}}
        \put(85,30){\makebox(0,0)[l]{$y'$}}
        \put(55,30){\makebox(0,0)[r]{$y$}}
        \put(-5,30){\makebox(0,0)[r]{$x'$}}
        \put(40,65){\makebox(0,0)[b]{$1$}}

        \put(40,0){\circle*{3}}
        \put(0,30){\circle*{3}}
        \put(20,30){\circle*{3}}
        \put(60,30){\circle*{3}}
        \put(80,30){\circle*{3}}
        \put(40,60){\circle*{3}}
      \end{picture}
    } 

  \end{picture}
  \caption{Greechie diagram for the lattice MO2 and its Hasse diagram.  The
dashed line indicates that the unconnected blocks belong to the same
Greechie diagram.
\label{fig:greechie-3}}
\end{figure}
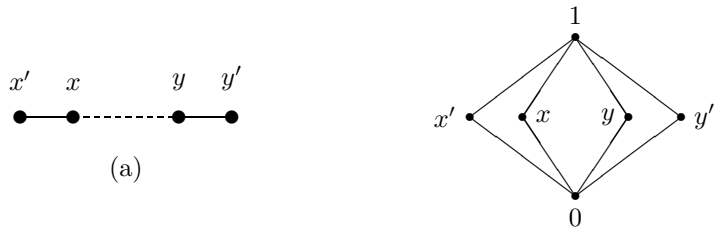

\section{Geometry: Generalized Orthoarguesian Equations}
\label{sec:goe}

Before 1975, the orthomodular lattice (OML) equations were the only ones
that were known to hold in a Hilbert lattice.  These have been
extensively studied in a vast body of research papers and books,
particularly in the context of the logic of quantum mechanics, and so
``orthomodular lattice'' and ``quantum logic'' have become almost
synonymous.

In 1975, Alan Day discovered an equation that holds in any
Hilbert lattice but does not in all OMLs.~\cite{gr-non-s}
He derived the equation, called the {\em orthoarguesian law}, by
imposing weakening orthogonality hypotheses on the so-called
Arguesian law, an equation closely related to the famous law of
projective geometry discovered by Desargues in the
1600's.\footnote{as part of an effort to help
artists, stonecutters, and engineers}

In 2000, Megill and Pavi{\v c}i{\'c} discovered a new infinite class
of equations that hold in any Hilbert lattice,\footnote{and
therefore in any infinite-dimensional Hilbert space}
called {\em generalized orthoarguesian equations} or $n$OA laws,
$n=3,4,\dots <\infty$,
a special case of which is the orthoarguesian law for $n=4$.

We recall the following definitions for reference here and later:
$a\to b\ {\buildrel\rm def\over =}\ a'\cup(a\cap b)$ and
and $a\equiv b\ {\buildrel\rm def\over =}\ (a\cap b)\cup(a'\cap b')$.

\begin{definition}
\label{def:noa}
We define an operation
${\buildrel (n)\over\equiv}$ on $n$ variables
$a_1,\ldots,a_n$ ($n\ge 3$) as follows:\footnote{To obtain
${\buildrel (n)\over\equiv}$ we substitute in each
${\buildrel (n-1)\over\equiv}$ subexpression only the two explicit
variables, leaving the other variables the same.  For example,
$(a_2{\buildrel (4)\over\equiv}a_5)$ on the right side of
(\ref{noaoper}) for $n=5$ means $(a_2{\buildrel
(3)\over\equiv}a_5)\cup ((a_2{\buildrel
(3)\over\equiv}a_4)\cap (a_5{\buildrel (3)\over\equiv}a_4))$
which means $(((a_2\to  a_3)\cap(a_5\to  a_3)) \cup((a_2'\to
a_3)\cap(a_5'\to a_3)))\cup
(
(((a_2\to  a_3)\cap(a_4\to  a_3))
\cup((a_2'\to  a_3)\cap(a_4'\to  a_3)))
\cap
(((a_5\to  a_3)\cap(a_4\to  a_3))
\cup((a_5'\to  a_3)\cap(a_4'\to  a_3)))
)
$.}
\begin{eqnarray}
a_1{\buildrel (3)\over\equiv}a_2\
&{\buildrel\rm def\over =}&\
((a_1\to  a_3)\cap(a_2\to  a_3))
\cup((a_1'\to  a_3)\cap(a_2'\to  a_3)) \\
a_1{\buildrel (n)\over\equiv}a_2\
&{\buildrel\rm def\over =}&\ (a_1{\buildrel (n-1)\over\equiv}a_2)\cup
((a_1{\buildrel (n-1)\over\equiv}a_n)\cap
(a_2{\buildrel (n-1)\over\equiv}a_n))\,,\quad n\ge 4\,.\label{noaoper}
\end{eqnarray}
\end{definition}

\begin{theorem}\label{th:noa}
The $n${\rm OA} {\em laws}
\begin{eqnarray}
(a_1\to a_3) \cap (a_1{\buildrel (n)\over\equiv}a_2)
\le a_2\to  a_3\,.\label{eq:noa}
\end{eqnarray}
hold in any Hilbert lattice.
\end{theorem}

\begin{proof}
To show that the $n$OA laws hold in ${\mathcal C}({\mathcal H})$, i.e., in a
Hilbert lattice, we closely follow the proof of the orthoarguesian
equation (4OA in our notation)  in Ref.~\cite{gr-non-s}. We recall
that in lattice ${\mathcal C}({\mathcal H})$, the meet corresponds to set
intersection and $\le$ to $\subseteq$.  We replace the join with
subspace sum $\mbox{\boldmath $+$}$ throughout:  the orthogonality
hypotheses permit us to do this on the left-hand side of the conclusion
\cite[Lemma 3 on p.~67]{kalmb83}, and on the right-hand side we use
$a\mbox{\boldmath $+$}b\subseteq a\cup b$.

In Ref.~\cite{mpoa99} we have shown that Eq.~(\ref{eq:noa}) can be
written as:
\begin{eqnarray}
\lefteqn{a_0 \perp b_0 \quad \&\quad a_1 \perp b_1
   \quad\&\quad\ldots\quad\&\quad a_{n-2}\perp b_{n-2}
   \quad \Rightarrow} \nonumber \\
& & ( a_0 \cup b_0 ) \cap ( a_1 \cup b_1 )\cap\cdots\cap (a_{n-2}\cup b_{n-2})
   \nonumber\\
& & \le b_0 \cup ( a_0 \cap ( a_1 \cup (
\cdots ( a_i \cup a_j ) \cap ( b_i \cup b_j )\cdots
 ) ) ), \label{eq:n-oa2}
\end{eqnarray}
where $a\perp b\ {\buildrel\rm def\over =}\ a\le b'$, $n\ge 3$,
 and $0\le i,j\le n-2$.
(The construction of the right-hand portion in ellipses can be inferred
starting from the 3OA basis, described in the next paragraph, and
building it up from $n$ to $n+1$ with the replacements described in
the last two sentences of this proof.)

The proof is by induction on $n$, starting at $n=3$.
Suppose $x$ is a vector belonging to the left-hand side of
(\ref{eq:n-oa2}).  Then there exist vectors $x_0\in a_0,\ y_0\in b_0,\
\ldots ,\ x_{n-2}\in a_{n-2},\ y_{n-2}\in b_{n-2}$ such that
$x=x_0+y_0=\cdots=x_{n-2}+y_{n-2}$.  Hence $x_k-x_l=y_l-y_k$ for $0\le
k,l\le n-2$.  In Eq.~(\ref{eq:n-oa2}) we assume, for our induction
hypothesis, that the components of vector $x=x_0+y_0$ can be distributed
over the leftmost terms on the right-hand side of the conclusion as
follows:
\[
  \cdots\subseteq
\underbrace{
  \underbrace{b_0}_{\textstyle y_0}
  \mbox{\boldmath $+$}(
  \underbrace{a_0}_{\textstyle x_0}
  \cap
  \underbrace{(
    \underbrace{a_1}_{\textstyle x_1}
    \mbox{\boldmath $+$}(
    \underbrace{
      \underbrace{(a_0\mbox{\boldmath $+$}a_1)}_{\textstyle x_0-x_1}
      \cap
      \underbrace{(b_0\mbox{\boldmath $+$}b_1)}_{\textstyle -y_0+y_1=x_0-x_1}
    }_{\textstyle x_0-x_1}
    \cap
    \underbrace{\cdots}_{\textstyle x_0-x_1}
    \cap
    \underbrace{\cdots}_{\textstyle x_0-x_1}
  }_{\textstyle x_1+(x_0-x_1)=x_0}
}_{\textstyle y_0+x_0=x}
\]
In particular, if we discard the right-hand ellipses, we obtain a ${\mathcal
C}({\mathcal H})$ proof of the 3OA law; this is the basis for our induction.

Let us first extend Eq.~(\ref{eq:n-oa2}) by adding variables $a_{n-1}$
and $b_{n-1}$
to the hypotheses and left-hand side of the conclusion.  The extended
Eq.~(\ref{eq:n-oa2}) so obtained obviously continues to hold in ${\mathcal
C}({\mathcal H})$.  Suppose $x$ is a vector belonging to the left-hand side
of this extended Eq.~(\ref{eq:n-oa2}).  Then there exist vectors $x_0\in
a_0,\ y_0\in b_0,\ \ldots ,\ x_{n-1}\in a_{n-1},\ y_{n-1}\in b_{n-1}$
such that
$x=x_0+y_0=\cdots=x_{n-1}+y_{n-1}$.  Hence $x_k-x_l=y_l-y_k$ for $0\le k,l\le
n-1$.  On the right-hand side of the extended Eq.~(\ref{eq:n-oa2}), for any
arbitrary subexpression of the form $( a_i \cup a_j ) \cap ( b_i \cup
b_j )$, where $i,j<n-1$, the vector components will be distributed
(possibly with signs reversed) as $x_i-x_j \in
a_i\mbox{\boldmath $+$} a_j$ and $x_i-x_j=-y_i+y_j \in b_i\mbox{\boldmath $+$}
b_j$.  If we replace $( a_i \cup a_j ) \cap ( b_i \cup b_j )$ with $(
a_i \cup a_j ) \cap ( b_i \cup b_j )\cap((( a_i \cup a_{n-1} ) \cap ( b_i
\cup b_{n-1} ))\cup(( a_j \cup a_{n-1} ) \cap ( b_j \cup b_{n-1} )))$,
components
$x_i$ and $x_j$ can be distributed as
\[
  \underbrace{( a_i \mbox{\boldmath $+$} a_j )}_{\textstyle x_i-x_j=}
  \cap
  \underbrace{( b_i \mbox{\boldmath $+$} b_j)}_{\textstyle -y_i+y_j}
  \cap
  \underbrace{
    ((\!
    \underbrace{( a_i \mbox{\boldmath $+$} a_{n-1} )}_{\textstyle x_i-x_{n-1}=}
    \cap
    \underbrace{( b_i \mbox{\boldmath $+$} b_{n-1} )}_{\textstyle -y_i+y_{n-1}})
    \mbox{\boldmath $+$}(\!\!\!
    \underbrace{( a_j \mbox{\boldmath $+$} a_{n-1} )\!\!}_{\textstyle -x_j+x_{n-1}=}
    \cap
    \underbrace{( b_j \mbox{\boldmath $+$} b_{n-1} )}_{\textstyle y_j-y_{n-1}}
    ))}_{\textstyle (x_i-x_{n-1})+(-x_j+x_{n-1})=x_i-x_j}
\]
so that $x_i-x_j$ remains an element of the replacement subexpression.
We continue to replace all subexpressions of the form $( a_i \cup a_j )
\cap ( b_i \cup b_j )$, where $i,j<n-1$, as above until they are
exhausted, obtaining the $(n+1)$OA law:
\begin{eqnarray}
\lefteqn{a_0 \perp b_0 \quad \&\quad a_1 \perp b_1
   \quad\&\quad\ldots\quad\&\quad a_{n-1}\perp b_{n-1}
   \quad \Rightarrow} \nonumber \\
& & ( a_0 \cup b_0 ) \cap ( a_1 \cup b_1 )\cap\cdots\cap (a_{n-1}\cup b_{n-1})
   \nonumber\\
& & \le b_0 \cup ( a_0 \cap ( a_1 \cup (
\cdots ( a_i \cup a_j ) \cap ( b_i \cup b_j
))\nonumber\\
& & \qquad\cap((( a_i \cup a_{n-1} ) \cap ( b_i \cup b_{n-1} ))\cup(( a_j \cup
a_{n-1} ) \cap ( b_j \cup b_{n-1} )))\cdots
 ) ) )\,.\qquad\label{eq:n-oa3}
\end{eqnarray}
\end{proof}

\begin{corollary}\label{th:gogrday-oa}
In any {\rm OML}, Day's orthoarguesian law {\rm \cite{gr-non-s}}
is equivalent to the {\rm 4OA} law and the equations found by
Godowski and Greechie in 1984 {\rm \cite{go-gr}} are equivalent to
each other and to {\rm 3OA}.
\end{corollary}

\begin{proof} As given in Ref.~\cite{mpoa99}.
\end{proof}

\begin{theorem}\label{th:oa-oml}Any ortholattice
{\rm (OL)~\cite[Def.~1]{pm-ql-l-hql1}} in which an 
$n${\rm OA} law holds is orthomodular.
No {\rm $n$OA} law holds in all {\em OML}s.
\end{theorem}

\begin{proof} All $n$OA laws fail in ortholattice {\rm O6}
({\em benzene ring, hexagon}) \cite[Sec.~2]{pm-ql-l-hql1}.

We prove the second statement of the theorem by finding an
orthomodular lattice in which the 3OA law fails. In Figure
\ref{fig:oa4f-5f} we show the smallest such Greechie diagram,
containing 13 atoms.  Since the $(n+1)$OA law implies
the $n$OA law (see Theorem~\ref{th:stronger} below),
the result follows.
\end{proof}

We conjecture that the second statement of the following theorem
holds for any $n$. To prove it for $n\ge6$ is an open problem.

\bigskip

\begin{theorem}\label{th:stronger}
In an {\rm OL}, the $n${\rm OA} law implies the $(n-1)${\rm OA}
law for any $n>3$. In an {\rm OL}, the $n${\rm OA} law does
not imply the $(n+1)${\rm OA} law for $3\le n\le 5$.
\end{theorem}
\begin{proof}
The first statement easily follows from the definition of the
$n$OA laws.

As for the second statement, we have three cases. For $n=3$,
the 3OA law holds in the 17-10-oa3p4f given in Fig.~\ref{fig:oa4f-5f}
and the 4OA law fails. For $n=4$, the 4OA law holds in 22-13-oa4p5f
given in the same figure, but the 5OA law
fails.\footnote{The notation ``35-23-oa5p6f''
means ``35 atoms, 23 edges, in which the 5OA law passes and the 6OA
law fails.''}

\begin{figure}[hbt]
\includegraphics[width=0.9999\textwidth]{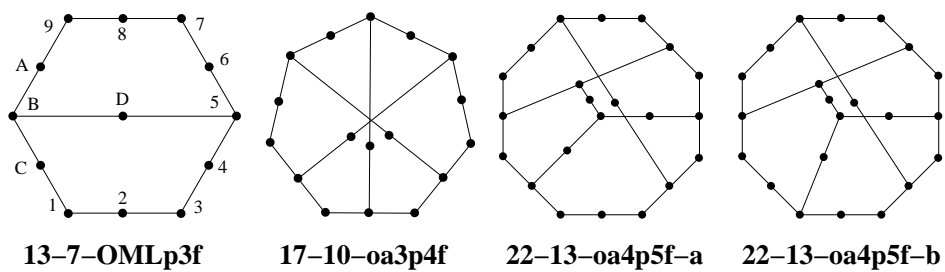}
\caption{The smallest Greechie diagram in which the OML law holds
and the 3OA law fails, the 3OA law holds and the 4OA fails, and
the two smallest Greechie diagrams in which 4OA holds and 5OA
fails. Cf.~Figs.~8$\>$b and 9 of Ref.~\cite{mpoa99}. McKay, Megill, and
Pavi{\v c}i{\'c} also introduced a textual way of writing down
Greechie diagrams that is self-explanatory for 13-7-OMLp-oa3f in
the figure: {\tt 123,345,567,789,9AB,BC1,BD5.}}
\label{fig:oa4f-5f}
\end{figure}

\begin{figure}[hbt]
\includegraphics[width=0.9999\textwidth]{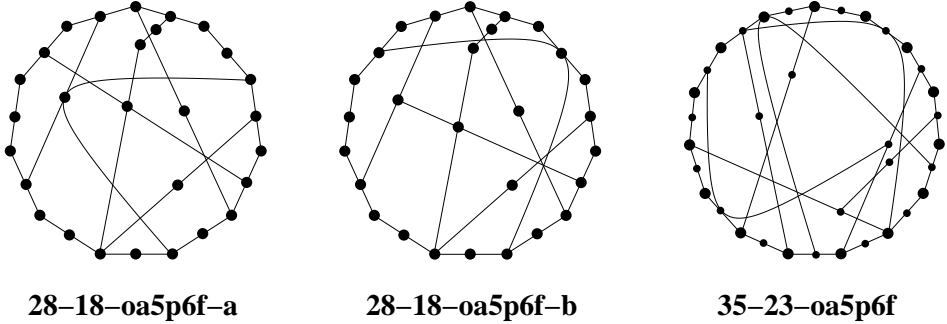}
\caption{Three lattices in which the 5OA law holds and the 6OA
fails. The 28 atom ones are apparently examples of the smallest
such lattices. The 35 atom one is from a set of lattices we
conjecture to contain the smallest lattices in which the 6OA law
holds and the 7OA fails. (Bigger dots denote the vertices of
the polygon.)}
\label{fig:oa6f}
\end{figure}

For $n=5$, the 5OA law holds in 28-18-oa5p6f (Fig.~\ref{fig:oa6f})
but the 6OA law fails.

These counterexamples were found using a program written by
Brendan McKay that exhaustively generates finite OML lattices,
that in turn fed a program written by Norman Megill that tests
the $n$OA laws against those lattices.~\cite{bdm-ndm-mp-1}
The $n$OA laws are very long equations whose lengths grow
exponentially with $n$ (with $4 \cdot 3^{n-2}+3$ variable
occurrences when expanded to elementary operations).
As $n$ increases, the difficulty of finding these
counter\-examples increases exponentially. Finding
counterexamples required over 10 years of CPU time on the
Cluster Isabella (224 CPUs) and Civil Enginering Cluster
(60 CPUs) of the University of Zagreb. Some additional
lattices in which 5OA holds and 6OA are:

35-23-oa5p6f:

{\parindent=0pt
\scriptsize{\tt

FTV$\!$,7TY$\!$,12Z$\!$,59Z$\!$,LMX$\!$,AJL$\!$,MUY$\!$,CLP$\!$,1AS$\!$,AGV$\!$,3EV$\!$,16K$\!$,2OX$\!$,DIV$\!$,58R$\!$,6HI$\!$,8AQ$\!$,7PR$\!$,CHN$\!$,7OW$\!$,9BV$\!$,4DU$\!$,ENO$\!$.

DHN$\!$,CDY$\!$,MQS$\!$,18B$\!$,LSV$\!$,ACL$\!$,5CO$\!$,3CK$\!$,6EK$\!$,79C$\!$,1LX$\!$,9MU$\!$,4CI$\!$,4PR$\!$,FJY$\!$,8FU$\!$,28Z$\!$,EPW$\!$,16G$\!$,LTW$\!$,RSZ$\!$,GOQ$\!$,2HW$\!$.

CHX$\!$,24Y$\!$,18S$\!$,8FT$\!$,6GN$\!$,17E$\!$,68Q$\!$,9BF$\!$,8DI$\!$,14V$\!$,8JW$\!$,1HP$\!$,AOX$\!$,IKR$\!$,26M$\!$,7GR$\!$,7UZ$\!$,4AL$\!$,3TY$\!$,3KP$\!$,79O$\!$,5GL$\!$,CJZ$\!$.

123$\!$,145$\!$,267$\!$,489$\!$,6AB$\!$,8AC$\!$,1DE$\!$,ADF$\!$,2GH$\!$,4IJ$\!$,FGI$\!$,BKL$\!$,KMN$\!$,MOP$\!$,OQR$\!$,GLQ$\!$,RST$\!$,FUV$\!$,UWX$\!$,7SW$\!$,8YZ$\!$,KVY$\!$,5NW$\!$.

123$\!$,145$\!$,267$\!$,489$\!$,6AB$\!$,8CD$\!$,AEF$\!$,3CE$\!$,1GH$\!$,GIJ$\!$,IKL$\!$,68K$\!$,LMN$\!$,MOP$\!$,EJO$\!$,HQR$\!$,AST$\!$,9QS$\!$,RUV$\!$,4WX$\!$,7UW$\!$,VYZ$\!$,4MY$\!$.}

}

36-24-oa5p6f:

{\parindent=0pt
\scriptsize{\tt 4IY$\!$,14C$\!$,48Z$\!$,BNW$\!$,CUa$\!$,9PU$\!$,DHX$\!$,BEX$\!$,68S$\!$,2OX$\!$,37D$\!$,35G$\!$,9NY$\!$,7AY$\!$,CHJ$\!$,PVX$\!$,8MP$\!$,45T$\!$,GKL$\!$,5FO$\!$,KPQ$\!$,OSa$\!$,LRY$\!$,6EG$\!$.

3MW$\!$,3OQ$\!$,DTV$\!$,FNR$\!$,OPa$\!$,9NS$\!$,7KS$\!$,5EZ$\!$,MRZ$\!$,HKT$\!$,1EU$\!$,45J$\!$,CIJ$\!$,BDM$\!$,3IY$\!$,3SX$\!$,6DP$\!$,8DL$\!$,24D$\!$,49G$\!$,8FY$\!$,4AQ$\!$,17A$\!$,HRa$\!$.

1FW$\!$,2AL$\!$,4DK$\!$,CNO$\!$,8KN$\!$,8FX$\!$,5Ca$\!$,1CE$\!$,ADV$\!$,3NZ$\!$,3GY$\!$,7HQ$\!$,9HI$\!$,4Ma$\!$,25S$\!$,2IT$\!$,9UZ$\!$,RTW$\!$,ABZ$\!$,DGR$\!$,6IK$\!$,1AJ$\!$,9Pa$\!$,1QY$\!$.

123$\!$,145$\!$,267$\!$,489$\!$,68A$\!$,1BC$\!$,BDE$\!$,AFG$\!$,FHI$\!$,6JK$\!$,DHJ$\!$,4LM$\!$,LNO$\!$,2IN$\!$,OPQ$\!$,PRS$\!$,1GR$\!$,QTU$\!$,TVW$\!$,JXY$\!$,NVX$\!$,5Za$\!$,SVZ$\!$,9ET$\!$.

123$\!$,145$\!$,267$\!$,489$\!$,68A$\!$,1BC$\!$,BDE$\!$,AFG$\!$,FHI$\!$,6JK$\!$,DHJ$\!$,4LM$\!$,LNO$\!$,2IN$\!$,OPQ$\!$,PRS$\!$,1GR$\!$,2TU$\!$,TVW$\!$,NXY$\!$,CVX$\!$,9Za$\!$,JYZ$\!$,STa$\!$.}

}

39-26-oa5p6f:

{\parindent=0pt
\scriptsize{\tt
bcd$\!$,YZa$\!$,VWX$\!$,TUX$\!$,RSd$\!$,QUd$\!$,OPa$\!$,NWa$\!$,LMa$\!$,KVc$\!$,IJS$\!$,FGH$\!$,EHP$\!$,JKZ$\!$,CDE$\!$,ABN$\!$,9ad$\!$,GMT$\!$,8BR$\!$,8DM$\!$,7IO$\!$,56X$\!$,46R$\!$,35O$\!$,2FW$\!$,17U$\!$.

bcd$\!$,YZa$\!$,VWX$\!$,UXa$\!$,RST$\!$,OPQ$\!$,LMN$\!$,IJK$\!$,FGH$\!$,DEN$\!$,BCT$\!$,AXd$\!$,89Z$\!$,7EZ$\!$,56H$\!$,67c$\!$,KQW$\!$,9MW$\!$,46P$\!$,OSa$\!$,3La$\!$,2Ga$\!$,1Ja$\!$,5CX$\!$,FId$\!$,8Rd$\!$.

bcd$\!$,YZa$\!$,WXa$\!$,UVa$\!$,RST$\!$,OPQ$\!$,NQT$\!$,MSd$\!$,KLd$\!$,IJa$\!$,HLa$\!$,EFG$\!$,DMV$\!$,BCZ$\!$,APY$\!$,89L$\!$,7Ac$\!$,679$\!$,8CT$\!$,7GX$\!$,5DP$\!$,4JQ$\!$,3FI$\!$,26S$\!$,1IR$\!$,BEd$\!$.

bcd$\!$,YZa$\!$,WXa$\!$,UVa$\!$,RST$\!$,PQT$\!$,MNO$\!$,KLd$\!$,JOd$\!$,HId$\!$,FGa$\!$,ENT$\!$,DIT$\!$,CMX$\!$,ABV$\!$,89Q$\!$,9Zc$\!$,7LY$\!$,6Ta$\!$,5BO$\!$,4CI$\!$,34G$\!$,58G$\!$,2AK$\!$,3KS$\!$,12N$\!$.

bcd$\!$,Zad$\!$,WXY$\!$,TUV$\!$,QRS$\!$,NOP$\!$,LMY$\!$,KPV$\!$,JPS$\!$,HIX$\!$,EFG$\!$,CDc$\!$,BIO$\!$,APY$\!$,9Yd$\!$,BDM$\!$,8GN$\!$,7Na$\!$,56U$\!$,CFS$\!$,47R$\!$,6HR$\!$,38M$\!$,2LU$\!$,14V$\!$,5Ed$\!$.}

}

To pursue the search for higher $n$'s would be too costly with
the present algorithms and classical computers.
\end{proof}

An interesting law that holds in an $n$OA lattice is the {\em $n${\rm
OA} identity law} given by the following Theorem.

\begin{theorem}\label{th:oa-equiv} In any {\rm $n$OA} we have:
\begin{eqnarray}
a_1{\buildrel (n)\over\equiv}a_2=1\qquad \Leftrightarrow\qquad
a_1\to a_n=a_2\to a_n
\label{eq:id1}
\end{eqnarray}
This also means that $a_1{\buildrel (n)\over\equiv}a_2$ being
equal to one is a relation of equivalence.
\end{theorem}

\begin{proof} As given in Ref.~\cite{mpoa99}.
\end{proof}

An immediate consequence of Eq.~(\ref{eq:id1}) is the transitive law
\begin{eqnarray}
a_2=b_1\qquad \&\qquad b_2=c_1\qquad \&\qquad c_2=a_1
  \qquad \&\qquad \nonumber\\
a_1{\buildrel (n)\over\equiv}a_2=1\qquad \&\qquad
b_1{\buildrel (n)\over\equiv}b_2=1\qquad \Rightarrow\qquad
c_1{\buildrel (n)\over\equiv}c_2=1
\label{eq:tr1}
\end{eqnarray}
that, while weaker than the $n$OA laws (verified to be strictly
weaker for $n=3,4$), nonetheless cannot be derived from the OML
axioms. \cite{mpoa99}.

The $n$OA identity laws bear a resemblance to the OML law in the form
$a\equiv b=1\Leftrightarrow a=b$.  Thus is it natural to think that they
might be equivalent to the $n$OA laws.  This is known as the {\em
orthoarguesian identity conjecture}, which asks whether the $n$OA
laws can be derived, in an OML, from Eq.~(\ref{eq:id1}).
Tests run against several million finite lattices (for $n=3$) have not
found a counterexample, but the conjecture has so far defied attempts to
find a proof.

An affirmative answer to this conjecture would provide us with a
powerful tool to prove new equivalents to the $n$OA laws.  It turns out
that it is often much easier to derive the $n$OA identity law from a
conjectured $n$OA law equivalent than it is to derive the $n$OA law
itself.  For example, under the assumption that the 3OA identity law
implies the 3OA law, all of the following conditions would be
established as equivalents to the 3OA law (where $aCb$ means $a=(a\cup
b)\cap(a\cup b')$ i.e.  $a$ commutes with $b$):
\begin{eqnarray}
(a_1\to  a_3)\cap (a_1{\buildrel (3)\over\equiv}a_2)  & C &  a_2\to  a_3
                  \label{eqa}\\
(a_1\to  a_3)\cap (a_1{\buildrel (3)\over\equiv}a_2)  & C &
   (a_2\to  a_3)\cap (a_1{\buildrel (3)\over\equiv}a_2)    \label{eqb}\\
(a_1'\to  a_3)'\cap (a_1{\buildrel (3)\over\equiv}a_2)  & \le &  a_2\to  a_3
                    \label{eqc}\\
(a_1'\to  a_3)'\cap (a_1{\buildrel (3)\over\equiv}a_2)  & C &
    a_2\to  a_3                \label{eqd}\\
(a_1'\to  a_3)'\cap (a_1{\buildrel (3)\over\equiv}a_2)  & C &
    (a_2\to  a_3)\cap (a_1{\buildrel (3)\over\equiv}a_2)  \label{eqe}\\
a_3\cap (a_1\to  a_3)\cap (a_1{\buildrel (3)\over\equiv}a_2)  & \le &
    (a_2\to  a_3)              \label{eqf}\\
a_3\cap (a_1\to  a_3)\cap (a_1{\buildrel (3)\over\equiv}a_2)  & C &
   (a_2\to  a_3)          \label{eqg}\\
a_3\cap (a_1\to  a_3)\cap (a_1{\buildrel (3)\over\equiv}a_2)  & C &
     (a_2\to  a_3)\cap (a_1{\buildrel (3)\over\equiv}a_2)     \label{eqh}\\
((a_1\to  a_3)\cap (a_1{\buildrel (3)\over\equiv}a_2))\to  a_3 & = &
    ((a_2\to  a_3)\cap (a_1{\buildrel (3)\over\equiv}a_2))\to  a_3
    \label{eqi}\\
((a_1\to  a_3)\cap (a_1{\buildrel (3)\over\equiv}a_2))\to  a_3  & C &
    ((a_2\to  a_3)\cap (a_1{\buildrel (3)\over\equiv}a_2))\to  a_3
  \label{eqj}
\end{eqnarray}
At the present time, only Eqs.~(\ref{eqa}) and (\ref{eqi}) from the
above set of conditions are known to be equivalent to the 3OA law.
Denoting the 3OA law [Eq.~(\ref{eq:noa}) for $n=3$] and the 3OA
identity law [Eq.~(\ref{eq:id1})] by OA3 and OI3 respectively, the
currently known relationships among the above conditions are as follows.
(Note that $\Rightarrow$ means ``the right-hand equation
can be proved from the axiom system of
OML $+$ the left-hand equation added as an axiom.'')
\begin{eqnarray}
&&\mbox{OA3 $\Leftrightarrow$ Eq.~(\ref{eqa}) $\Rightarrow$ Eq.~(\ref{eqb})
   $\Rightarrow$ OI3} \nonumber \\
&&\mbox{OA3 $\Rightarrow$ Eq.~(\ref{eqc}) $\Rightarrow$ Eq.~(\ref{eqd})
    $\Rightarrow$ Eq.~(\ref{eqe})  $\Rightarrow$ OI3}\nonumber \\
&&\mbox{OA3 $\Rightarrow$ Eq.~(\ref{eqf}) $\Leftrightarrow$
    Eq.~(\ref{eqg}) $\Leftrightarrow$ Eq.~(\ref{eqh})
     $\Rightarrow$ OI3}\nonumber \\
&&\mbox{OA3 $\Leftrightarrow$ Eq.~(\ref{eqi})
  $\Rightarrow$ Eq.~(\ref{eqj})  $\Rightarrow$ OI3}\nonumber
\end{eqnarray}

\section{States: Godowski Equations}
\label{sec:ge}

As we explained in Section \ref{sec:hl}, there is a way to obtain
complex infinite-dimensional Hilbert space from the Hilbert
lattice equipped with several additional conditions and without
invoking the notion of state at all. States then follow by Gleason's
theorem (see Theorem \ref{th:gleason} in Section \ref{sec:d}). However,
we can also define states directly on an ortholattice, and then it
turns out that such a definition generates many properties of the
lattice that hold in any Hilbert lattice. In particular, the states
generate the Godowski and Mayet-Godowski equations (on which we will
elaborate in the next section).

\begin{definition}\label{def:state} A {\em state} (also called
{\em probability measures} or simply {\em probabilities}
\cite{kalmb98,kalmb83,kalmb86,kalmb98,maczin})
on a lattice $\mathcal L$
 is a function $m:{\mathcal L}\longrightarrow [0,1]$
such that $m(1)=1$ and $a\perp b\ \Rightarrow\ m(a\cup b)=m(a)+m(b)$,
where $a\perp b$ means $a\le b'$.
\end{definition}

\begin{lemma}\label{lem:state}
The following properties hold for any state $m$:
\begin{eqnarray}
&m(a)+m(a')=1\label{eq:state3}\\
&a\le b\ \Rightarrow\ m(a)\le m(b)\label{eq:state4}\\
&0\le m(a)\le 1\label{eq:state5}\\
&m(a_1)=\cdots=m(a_n)=1\
\Leftrightarrow\ m(a_1)+\cdots+m(a_n)=n\label{eq:state6}\\
&m(a_1\cap\cdots\cap a_n)=1\ \Rightarrow\
m(a_1)=\cdots=m(a_n)=1\label{eq:state7}
\end{eqnarray}
\end{lemma}

\begin{definition}\label{def:strong} A nonempty set $S$ of
states on $\mathcal L$ is called a strong set of {\em classical\/}
states if
\begin{eqnarray}
(\exists m \in S)(\forall a,b\in{\mathcal L})((m(a)=1\ \Rightarrow
\ m(b)=1)\ \Rightarrow\ a\le b)\,\quad
\label{eq:st-cl}
\end{eqnarray}
and a strong set of {\em quantum\/} states if
\begin{eqnarray}
(\forall a,b\in{\rm L})(\exists m \in S)((m(a)=1\ \Rightarrow
\ m(b)=1)\ \Rightarrow\ a\le b)\,.\quad
\label{eq:st-qm}
\end{eqnarray}
\end{definition}

We want to emphasize the difference between quantum and classical
states. A classical state is the same for all lattice elements,
while a quantum state might be different for each of the
elements. The following theorem \cite{mpoa99} shows us that
a classical state can be be very strong.

\begin{theorem}\label{th:strong-distr} Any ortholattice that admits a
strong set of classical states is distributive.
\end{theorem}

In 1981, Radoslaw Godowski \cite{godow} found an infinite series of
equations partly corresponding to the strong set of quantum states given
by Eq.~(\ref{eq:st-qm}), forming a series of algebras contained in the
class of all orthomodular lattices and containing the class of all
Hilbert lattices.  Importantly, there are OMLs that do not admit a
strong set of states, so Godowski's equations provide us with new
equational laws that extend the OML laws that hold in Hilbert lattices.

Before deriving the equations themselves, we will first prove, directly
in Hilbert space, that Hilbert lattices admit strong sets of states.
This will provide some insight into how these equations arise.

\begin{theorem}\label{hilb-strong-s} Any Hilbert lattice admits a
strong set of states.
\end{theorem}
\begin{proof} We need only to use pure states defined by unit vectors:
If $a$ and $b$ are closed subspaces of a Hilbert space $\mathcal H$
such that $a$ is not contained in $b$, there is a unit vector $u$ of
$\mathcal H$ belonging to $a-b$. If for each $c$ in the lattice of all
closed subspaces of $\mathcal H$, ${\mathcal C}({\mathcal H})$,
we define $m(c)$ as the square of the norm of
the projection of $u$ onto $c$, then $m$ is a state on $\mathcal H$ such
that $m(a)=1$ and $m(b)<1$. This proves that ${\mathcal C}({\mathcal H})$
admits a strong set of states, and this proof works in each of the
3 cases where the underlying field is the field of real numbers,
of complex numbers, or of quaternions.

We can formalize the proof as follows:
\begin{eqnarray}
&&(\forall a,b\in L)((\sim\ a\le b)\ \Rightarrow\
(\exists m\in S)(m(a)=1\ \&\ \sim\ m(b)=1))\nonumber\\
&\Rightarrow & (\forall a,b\in L)(\exists m\in S)((m(a)=1\
\Rightarrow\ m(b)=1)\ \Rightarrow\ a\le b)\nonumber
\end{eqnarray}
\end{proof}

We will now define the family of equations found by Godowski,
introducing a special notation for them.  Then we will prove
that they hold in any lattice admitting a strong set of states
and thus, in particular, any Hilbert lattice.

\begin{definition}\label{def:god-equiv}Let us call the
following expression the {\em Godowski identity}:
\begin{eqnarray}
a_1{\buildrel\gamma\over\equiv}a_n{\buildrel{\rm def}
\over =}(a_1\to a_2)\cap(a_2\to a_3)\cap\cdots
\cap(a_{n-1}\to a_n)\cap(a_n\to a_1),\nonumber\\
\qquad \qquad \qquad \qquad \qquad
n=3,4,\dots\label{eq:god-equiv}
\end{eqnarray}
\end{definition}
We define $a_n{\buildrel\gamma\over\equiv}a_1$ in the same way with
variables $a_i$ and $a_{n-i+1}$ swapped;
in general $a_i{\buildrel\gamma\over\equiv}a_j$ will be an expression
with $|j-i|+1\ge 3$ variables $a_i,\ldots,a_j$ first appearing in that
order.  For completeness and later use (Theorem \ref{th:god-trans}) we
define $a_i{\buildrel\gamma\over\equiv}a_i{\buildrel{\rm def} \over
=} (a_i\to  a_i)=1$ and
$a_i{\buildrel\gamma\over\equiv}a_{i+1}{\buildrel{\rm def} \over
=}(a_i\to  a_{i+1})\cap(a_{i+1}\to  a_i)=a_i\equiv a_{i+1}$, the last
equality holding in any OML.

\begin{theorem}\label{th:god-eq} Godowski's equations {\em\cite{godow}}
\begin{eqnarray}
a_1{\buildrel\gamma\over\equiv}a_3
&=&a_3{\buildrel\gamma\over\equiv}a_1
\label{eq:godow3o}\\
a_1{\buildrel\gamma\over\equiv}a_4
&=&a_4{\buildrel\gamma\over\equiv}a_1
\label{eq:godow4o}\\
a_1{\buildrel\gamma\over\equiv}a_5
&=&a_5{\buildrel\gamma\over\equiv}a_1
\label{eq:godow5o}\\
&\dots &\nonumber
\end{eqnarray}
hold in all ortholattices, {\em OL}'s, with strong sets of states.
An {\em OL} to which these equations are added is a variety
smaller than {\em OML}.

We shall call these equations {\rm $n$-Go} {\rm (}{\rm 3-Go},
{\rm 4-Go}, etc.\/{\rm )}.  We also denote by
{\rm $n$GO} {\rm (}{\rm 3GO}, {\rm 4GO}, etc.\/{\rm )} the
{\rm OL} variety determined by {\rm $n$-Go} and
call it the {\rm $n$GO law}.
\end{theorem}
\begin{proof} By
Definition \ref{def:state} we have
$m(a_1\to a_2)=m(a'_1)+m(a_1\cap a_2)$ etc.,
because $a'_1\le(a'_1\cup a'_2)$, i.e., $a'_1\perp(a_1\cap a_2)$
in any ortholattice. Assuming $m(a_1{\buildrel\gamma\over\equiv}a_n)=1$,
we have $m(a_1\to a_2)=\cdots =m(a_{n-2}\to a_n)=m(a_n\to a_1)=1$.
Hence, $n=m(a_1\to a_2)+\cdots+m(a_{n-2}\to a_n)+m(a_n\to a_1)=
m(a_n\to a_{n-2})+\cdots+m(a_2\to a_1)+m(a_1\to a_n)$.
This last equality follows from breaking up, rearranging,
and recombining the $m(a_i\to  a_j)$
terms as described by the first sentence. Therefore,
$m(a_n\to a_{n-2})=\cdots=m(a_2\to a_1)=m(a_1\to a_n)=1$. Thus,
by Definition \ref{def:strong} for strong quantum states, we obtain:
$(a_1{\buildrel\gamma\over\equiv}a_n)\le (a_n\to a_{n-2})$,\ \ldots,
$(a_1{\buildrel\gamma\over\equiv}a_n)\le (a_2\to a_1)$, and
$(a_1{\buildrel\gamma\over\equiv}a_n)\le (a_1\to a_n)$, wherefrom
we get $(a_1{\buildrel\gamma\over\equiv}a_n)\le
(a_n{\buildrel\gamma\over\equiv}a_1)$. By symmetry, we get
$(a_n{\buildrel\gamma\over\equiv}a_1)\le
(a_1{\buildrel\gamma\over\equiv}a_n)$. Thus
$(a_1{\buildrel\gamma\over\equiv}a_n)=
(a_n{\buildrel\gamma\over\equiv}a_1)$.

$n$GO implies the orthomodular law
 because 3-Go fails in O6, and $n$-Go implies
$(n-1)$-Go in any OL (Lemma \ref{lem:god-iimpliesn-1}). It is a variety
smaller than OML because 3-Go fails in the Greechie diagram of
Fig.~\ref{fig:oag34}a.
\end{proof}

\begin{lemma}\label{lem:god-iimpliesn-1}
Any {\em $n$GO} is an {\em $(n-1)$GO}, $n=4,5,6,\ldots$
\end{lemma}
\begin{proof}
Substitute $a_1$ for $a_2$ in equation $n$-Go.
\end{proof}

\begin{figure}[htbp]\centering
  \setlength{\unitlength}{0.6pt}
  \begin{picture}(240,120)(0,0)

    \put(-15,13) { 
      \begin{picture}(124,110)(0,0) 
        \put(32.2,0){\line(1,0){55.6}}
        \put(32.2,100){\line(1,0){55.6}}
        \put(2.3,50){\line(3,5){29.9}}
        \put(117.7,50){\line(-3,5){29.9}}
        \put(2.3,50){\line(3,-5){29.9}}
        \put(117.7,50){\line(-3,-5){29.9}}
        \put(60,100){\line(0,-1){50}}
        \put(17.25,25){\line(5,3){42.75}}
        \put(102.75,25){\line(-5,3){42.75}}

        \put(2.3,50){\circle*{5}}
        \put(117.7,50){\circle*{5}}
        \put(87.8,0){\circle*{5}}
        \put(87.8,100){\circle*{5}}
        \put(32.2,0){\circle*{5}}
        \put(32.2,100){\circle*{5}}
        \put(60,0){\circle*{5}}
        \put(60,100){\circle*{5}}
        \put(17.25,25){\circle*{5}}
        \put(17.25,75){\circle*{5}}
        \put(102.75,25){\circle*{5}}
        \put(102.75,75){\circle*{5}}
        \put(60,50){\circle*{5}}
        \put(38.9,37.5){\circle*{5}}
        \put(81.4,37.5){\circle*{5}}
        \put(60,75){\circle*{5}}
      \end{picture}
    } 

    \put(150,0) { 
      \begin{picture}(124,110)(0,0) 
        \put(35.15,0){\line(1,0){49.7}}
        \put(35.15,120){\line(1,0){49.7}}
        \put(0,35.15){\line(0,1){49.7}}
        \put(120,35.15){\line(0,1){49.7}}
        \put(0,35.15){\line(1,-1){35.15}}
        \put(0,84.85){\line(1,1){35.15}}
        \put(120,35.15){\line(-1,-1){35.15}}
        \put(120,84.85){\line(-1,1){35.15}}
        \put(60,0){\line(1,6){12}}
        \put(0,60){\line(6,1){72}}
        \put(60,120){\line(1,-4){12}}
        \put(120,60){\line(-4,1){48}}

        \put(35.15,0){\circle*{5}}
        \put(35.15,120){\circle*{5}}
        \put(84.85,0){\circle*{5}}
        \put(84.85,120){\circle*{5}}
        \put(0,35.15){\circle*{5}}
        \put(0,84.85){\circle*{5}}
        \put(120,35.15){\circle*{5}}
        \put(120,84.85){\circle*{5}}
        \put(60,0){\circle*{5}}
        \put(60,120){\circle*{5}}
        \put(0,60){\circle*{5}}
        \put(120,60){\circle*{5}}
        \put(17.575,17.575){\circle*{5}}
        \put(17.575,102.425){\circle*{5}}
        \put(102.425,17.575){\circle*{5}}
        \put(102.425,102.425){\circle*{5}}
        \put(72,72){\circle*{5}}
        \put(36,66){\circle*{5}}
        \put(66,36){\circle*{5}}
        \put(96,66){\circle*{5}}
        \put(66,96){\circle*{5}}

      \end{picture}
    } 

  \end{picture}
  \caption{(a)
   Greechie diagram for OML G3;
   (b) Greechie diagram for OML G4.
\label{fig:oag34}}
\end{figure}
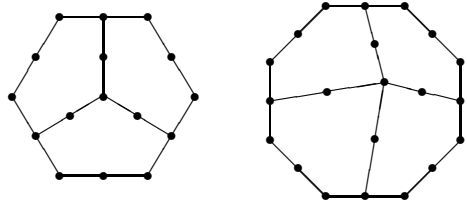

The converse of Lemma~(\ref{lem:god-iimpliesn-1}) does not hold.
Indeed, the {\it wagon wheel} OMLs G$n$, $n=3,4,5,\ldots$, are related to
the $n$-Go equations in the sense that G$n$ violates $n$-Go but (for
$n\ge 4$) not $(n-1)$-Go.  In Fig.~\ref{fig:oag34} we show examples G3
and G4; for larger $n$ we construct G$n$ by adding more ``spokes'' in
the obvious way (according to the general scheme described in
\cite{godow}).

Megill and Pavi\v ci\'c \cite{mpoa99} explored many properties and
consequences of the $n$-Go equations.  The theorems below, whose proofs
we omit and can be found in the cited reference, summarize some of the
results their work.

\begin{theorem}\label{th:god-th1} An {\em OL} in which any of
the following equations holds is an {\em $n$GO} and vice versa.
\begin{eqnarray}
a_1{\buildrel\gamma\over\equiv}a_n
\!\!&=&\!\!(a_1\equiv a_2)\cap(a_2\equiv a_3)\cap\cdots\cap(a_{n-1}\equiv a_n)
\label{eq:godow1c}\qquad\\
a_1{\buildrel\gamma\over\equiv}a_n
\!\!&\le&\!\!a_1\to a_n,
\label{eq:godow1d}\\
(a_1{\buildrel\gamma\over\equiv}a_n)\cap(a_1\cup a_2\cup\cdots\cup a_n)
\!\!&=&\!\!a_1\cap a_2\cap\cdots\cap a_n
\label{eq:godow1e}
\end{eqnarray}
\end{theorem}

\begin{theorem}\label{th:god-th2}
In any {\rm $n$GO}, $n=3,4,5,\ldots$, the following relations hold.
\begin{eqnarray}
a_1{\buildrel\gamma\over\equiv}a_n
&\le&a_j\to a_k, \qquad 1\le j\le n,\ 1\le k\le n
\label{eq:godow1d2}
\end{eqnarray}
\end{theorem}

The $n$-Go equations can be equivalently expressed as inferences
involving $2n$ variables, as the following theorem shows.  In this form
they can be useful for certain kinds of proofs.
\begin{theorem}\label{th:go2n} Any {\em OML} in which
\begin{eqnarray}
\lefteqn{a_1\perp b_1\perp a_2\perp b_2\perp\ldots\perp a_n\perp b_n\perp a_1
    \qquad\Rightarrow} & & \nonumber \\
& & (a_1\cup b_1)\cap (a_2\cup b_2)\cap\cdots\cap(a_n\cup b_n)\le b_1\cup a_2
\label{eq:go2n}
\end{eqnarray}
holds is an {\em $n$GO} and vice versa.
\end{theorem}

Finally, the following theorem shows a transitive-like property
that can be derived from the Godowski equations.

\begin{theorem}\label{th:god-trans}
The following equation holds in {\em $n$GO}, where $i,j\ge 1$ and
$n=\max(i,j,3)$.
\begin{eqnarray}
(a_1{\buildrel\gamma\over\equiv}a_i)
\cap(a_i{\buildrel\gamma\over\equiv}a_j)
&\le&a_1{\buildrel\gamma\over\equiv}a_j
\end{eqnarray}
\end{theorem}

While the wagon wheel OMLs characterize $n$GO equations in an elegant
way, they are not the smallest OMLs that are not $n$GOs.  Smaller OMLs
can be used to distinguish $n+1$-Go from $n$-Go, which can improve
computational efficiency.  For example, the Peterson OML, G4s,
Fig.~\ref{fig:oag6}$\>$(a), is the smallest that violates 4-Go but
not 3-Go; it has 32 nodes vs.\ 44 nodes in the wagon wheel G4 in
Fig.~\ref{fig:oag34}$\>$(b).
Lattice G5s, Fig.~\ref{fig:oag6}$\>$(b), with 42 nodes (vs.~54 nodes in
G5), is the smallest that violates 5-Go but not 4-Go.
OML G6s2, Fig.~\ref{fig:oag6}$\>$(c) is one of three smallest that
violates 6-Go but not 5-Go, with 44 nodes (vs.\ 64 nodes) in G6.
Lattice G7s1, Fig.~\ref{fig:oag6}$\>$(d), is one of several
smallest we obtained to violate 7-Go but not 6-Go. They both have
50 nodes, respectively (vs.~74 nodes in G7).

\begin{figure}[htbp]\centering
  \setlength{\unitlength}{0.8pt}
  \begin{picture}(250,115)(0,0)

\setlength{\unitlength}{0.7pt}
   \put(-115,13) { 
      \begin{picture}(124,110)(0,0) 
        \put(32.2,0){\line(1,0){55.6}}
        \put(32.2,100){\line(1,0){55.6}}
        \put(2.3,50){\line(3,5){29.9}}
        \put(117.7,50){\line(-3,5){29.9}}
        \put(2.3,50){\line(3,-5){29.9}}
        \put(117.7,50){\line(-3,-5){29.9}}
        \put(60,100){\line(0,-1){100}}
        \put(17.25,25){\line(5,3){84.8}}
        \put(102.75,25){\line(-5,3){84.8}}
        \put(34.71,65.6){\line(1,0){50.58}}

        \put(2.3,50){\circle*{5}}
        \put(117.7,50){\circle*{5}}
        \put(87.8,0){\circle*{5}}
        \put(87.8,100){\circle*{5}}
        \put(32.2,0){\circle*{5}}
        \put(32.2,100){\circle*{5}}
        \put(60,0){\circle*{5}}
        \put(60,100){\circle*{5}}
        \put(17.25,25){\circle*{5}}
        \put(17.85,76){\circle*{5}}
        \put(102.75,25){\circle*{5}}
        \put(102.15,76){\circle*{5}}
        \put(34.71,65.6){\circle*{5}}
        \put(85.29,65.6){\circle*{5}}
        \put(60,65.6){\circle*{5}}
      \end{picture}
    } 

\setlength{\unitlength}{0.64pt}
    \put(23,10) { 
      \begin{picture}(100,90)(0,0) 
        \put(39.48,3.62){\line(1,0){41.04}}
        \put(39.48,3.62){\line(-6,5){31.44}}
        \put(80.52,3.62){\line(6,5){31.44}}
        \put(8.04,30){\line(-1,6){6.75}}
        \put(111.96,30){\line(1,6){6.75}}
        \put(0.91,70.42){\line(3,5){21.2}}
        \put(119.09,70.42){\line(-3,5){21.2}}
        \put(22.13,105.96){\line(3,1){38}}
        \put(97.87,105.96){\line(-3,1){38}}
        \put(60,3.62){\line(0,1){114.98}}
        \put(40.715,112.3){\line(3,-5){56.7}}
        \put(79.285,112.3){\line(-3,-5){56.7}}
        \put(8.04,30){\line(5,3){99.5}}
        \put(111.96,30){\line(-5,3){99.5}}
        \put(4.475,50.21){\line(1,0){111.05}}

        \put(80.52,3.62){\circle*{5}}
        \put(39.48,3.62){\circle*{5}}
        \put(111.96,30){\circle*{5}}
        \put(8.04,30){\circle*{5}}
        \put(119.09,70.42){\circle*{5}}
        \put(0.91,70.42){\circle*{5}}
        \put(97.2,105.8){\circle*{5}}
        \put(21.8,105.8){\circle*{5}}
        \put(60,118.6){\circle*{5}}
        \put(60,3.62){\circle*{5}}
        \put(97.24,17.81){\circle*{5}}
        \put(22.76,17.81){\circle*{5}}
        \put(115.525,50.21){\circle*{5}}
        \put(4.475,50.21){\circle*{5}}
        \put(107.83,89.19){\circle*{5}}
        \put(12.17,89.19){\circle*{5}}
        \put(79.285,112.5){\circle*{5}}
        \put(40.715,112.5){\circle*{5}}
        \put(60,50.21){\circle*{5}}
        \put(51.3,66.6){\circle*{5}}
        \put(68.7,66.6){\circle*{5}}

      \end{picture}
    } 

\setlength{\unitlength}{0.6pt}
    \put(180,15) { 
      \begin{picture}(100,90)(0,0) 
        \put(35.15,0){\line(1,0){49.7}}
        \put(35.15,120){\line(1,0){49.7}}
        \put(0,35.15){\line(0,1){49.7}}
        \put(120,35.15){\line(0,1){49.7}}
        \put(0,35.15){\line(1,-1){35.15}}
        \put(0,84.85){\line(1,1){35.15}}
        \put(120,35.15){\line(-1,-1){35.15}}
        \put(120,84.85){\line(-1,1){35.15}}
        \put(60,120){\line(0,-1){120}}
        \put(0,60){\line(5,2){103.5}}
        \put(120,60){\line(-5,2){103.5}}
        \put(17,18.2){\line(2,3){67.8}}
        \put(103,18.2){\line(-2,3){67.8}}
        \put(60,0){\line(-1,4){23.3}}

        \put(35.15,0){\circle*{5}}
        \put(35.15,120){\circle*{5}}
        \put(84.85,0){\circle*{5}}
        \put(84.85,120){\circle*{5}}
        \put(0,35.15){\circle*{5}}
        \put(0,84.85){\circle*{5}}
        \put(120,35.15){\circle*{5}}
        \put(120,84.85){\circle*{5}}
        \put(60,0){\circle*{5}}
        \put(60,120){\circle*{5}}
        \put(0,60){\circle*{5}}
        \put(120,60){\circle*{5}}
        \put(17,18.2){\circle*{5}}
        \put(16.575,101.425){\circle*{5}}
        \put(103,18.2){\circle*{5}}
        \put(103.425,101.425){\circle*{5}}
        \put(60,38){\circle*{5}}
        \put(40.8,76.2){\circle*{5}}
        \put(36.7,93.2){\circle*{5}}
        \put(30,38){\circle*{5}}
        \put(90,38){\circle*{5}}
      \end{picture}
    } 

\setlength{\unitlength}{0.65pt}
    \put(300,10) { 
      \begin{picture}(100,90)(0,0) 
        \put(39.48,3.62){\line(1,0){41.04}}
        \put(39.48,3.62){\line(-6,5){31.44}}
        \put(80.52,3.62){\line(6,5){31.44}}
        \put(8.04,30){\line(-1,6){6.75}}
        \put(111.96,30){\line(1,6){6.75}}
        \put(0.91,70.42){\line(3,5){21.2}}
        \put(119.09,70.42){\line(-3,5){21.2}}
        \put(22.13,105.96){\line(3,1){38}}
        \put(97.87,105.96){\line(-3,1){38}}
        \put(12.17,89.19){\line(1,0){95.3}}
        \put(4.475,50.21){\line(5,3){93.1}}
        \put(4.475,50.21){\line(1,2){15.8}}
        \put(115.525,50.21){\line(-3,1){95}}
        \put(28,64){\line(3,-2){69}}
        \put(39.48,56){\line(5,3){55.3}}
        \put(39.48,3.62){\line(0,1){108}}
        \put(80.52,3.62){\circle*{5}}
        \put(39.48,3.62){\circle*{5}}
        \put(111.96,30){\circle*{5}}
        \put(8.04,30){\circle*{5}}
        \put(119.09,70.42){\circle*{5}}
        \put(0.91,70.42){\circle*{5}}
        \put(97.2,105.8){\circle*{5}}
        \put(21.8,105.8){\circle*{5}}
        \put(60,118.6){\circle*{5}}
        \put(60,3.62){\circle*{5}}
        \put(97.24,17.81){\circle*{5}}
        \put(22.76,17.81){\circle*{5}}
        \put(115.525,50.21){\circle*{5}}
        \put(4.475,50.21){\circle*{5}}
        \put(107.83,89.19){\circle*{5}}
        \put(12.17,89.19){\circle*{5}}
        \put(79.285,112.5){\circle*{5}}
        \put(39.48,112.3){\circle*{5}}
        \put(94.3,89.19){\circle*{5}}
        \put(60.4,68.5){\circle*{5}}
        \put(39.48,56){\circle*{5}}
        \put(28,64){\circle*{5}}
        \put(20.5,82){\circle*{5}}
        \put(13.4,67.6){\circle*{5}}
      \end{picture}
}
  \end{picture}
  \caption{\ (a) OML G4s; \quad(b)
    OML G5s; \quad (c) OML G6s;
\quad (d) OML G7s.\
\label{fig:oag6}}
\end{figure}
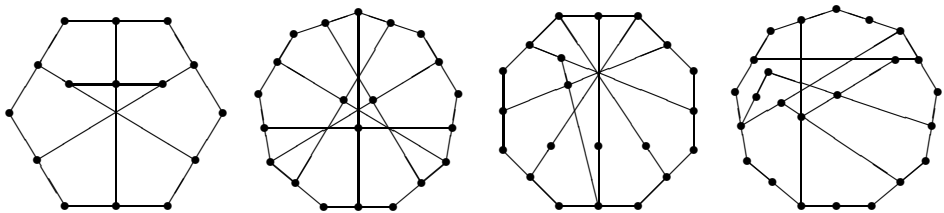

\section{States: Mayet-Godowski Equations}
\label{sec:mge}

In 1985, Ren\'e Mayet \cite{mayet85} described an equational variety of
lattices, which he called $OM_S^*$, that included all Hilbert lattices
and were included in the $n$GO varieties (found by
Godowski) that we described in the previous section.  In 1986, Mayet
\cite{mayet86} displayed several examples of equations that hold in this
new variety.  However, Megill and Pavi{\v c}i{\'c} \cite{mpoa99} showed
that all of Mayet's equational examples can be derived in $n$GO for some
$n$.  Thus it remained unclear whether Mayet's variety was strictly
contained in the $n$GOs.

In this section, we will show that Mayet's variety, which we will call
MGO, is indeed strictly contained in all $n$GOs
(Theorem~\ref{th:mgo-lt-ngo}).  We will do this by exhibiting an equation
that holds in his variety (and thus in all Hilbert lattices) but cannot
be derived in any $n$GO, following Megill and Pavi{\v c}i{\'c}
\cite{mp-gen-godowski06-arXiv}.

We will also describe a general family of equations that hold in all
Hilbert lattices and contains the new equation, and we will define a
simplified notation for representing these equations.

We call the equations in this family {\em Mayet-Godowski
equations} and, in Theorem~\ref{th:mge}, prove that they hold in all
Hilbert lattices.\footnote{A family of equations equivalent to the
family MGE, with a different presentation, was given by Mayet as
$E(Y_2)$ on p.~183 of \cite{mayet86}.}

\begin{definition}\label{def:gge}
A {\em Mayet-Godowski equation} ({\rm MGE}) is an equality
with $n\ge 2$ conjuncts on each side:
\begin{eqnarray}
t_1 \cap \cdots\cap t_n = u_1 \cap \cdots\cap u_n
\end{eqnarray}
where each conjunct $t_i$ (or $u_1$) is a term consisting of
either a variable or a disjunction of two or more distinct
variables:
\begin{eqnarray}
t_i = a_{i,1}\cup \cdots\cup a_{i,p_i}\qquad \mbox{i.e. $p_i$ disjuncts}\\
u_i = b_{i,1}\cup \cdots\cup b_{i,q_i}\qquad \mbox{i.e. $q_i$ disjuncts}
\end{eqnarray}
and where the following conditions are imposed on the set of variables
in the equation:
\begin{enumerate}
\item{All variables in a given term $t_i$ or $u_i$ are
      mutually orthogonal.}
\item{ Each variable occurs the same number of times on each side of
       the equality.}
\end{enumerate}
\end{definition}

We will call a lattice in which all MGEs hold an MGO; i.e., MGO is the
class (equational variety) of all lattices in which all MGEs hold.

\begin{lemma}\label{lem:mge1}
In any {\rm OL},
\begin{eqnarray}
a\perp b\quad \& \quad a\perp c \quad
   \Rightarrow \quad a\perp (b\cup c)\label{eq:mge1}
\end{eqnarray}
\end{lemma}
\begin{proof}
Trivial.
\end{proof}

\begin{lemma}\label{lem:mge2}
If $a_1,\ldots a_n$ are mutually orthogonal, then
\begin{eqnarray}
m(a_1)+\cdots+m(a_n) = m(a_1 \cup \cdots\cup a_n)\label{eq:mge2}
\end{eqnarray}
\end{lemma}
\begin{proof}

For $n=2$,  $a_1\perp a_2$ implies
$m(a_1\cup a_2)=m(a_1)+m(a_2)$
by Definition~\ref{def:state}.

For $n=3$,
$a_1\perp a_2$ and $a_1\perp a_3$ imply
$a_1\perp (a_2\cup a_3)$ by Lemma~\ref{lem:mge1}.  So by
Definition~\ref{def:state}, $m(a_1\cup (a_2\cup a_3))=
m(a_1)+m(a_2\cup a_3)$.  Again by Definition~\ref{def:state},
$a_2\perp a_3$ implies $m(a_2\cup a_3)=m(a_2)+m(a_3)$.

For any $n>2$, we apply the obvious induction step to the
$n-1$ case:
$m((a_1\cup\cdots\cup a_{n-1})\cup a_n)=
m(a_1\cup\cdots\cup a_{n-1})+m(a_n)=m(a_1)+\cdots+m(a_{n-1})+ m(a_n)$.
\end{proof}

\begin{theorem}\label{th:mge}
A Mayet-Godowski equation holds in any ortholattice
 $\mathcal L$
ad\-mit\-ting a strong set of states and thus, in particular,
in any Hilbert lattice.
\end{theorem}
\begin{proof}
Suppose that for some state $m$,  $m(t_1 \cap\cdots\cap t_n) = 1$.
Then by Eq.~(\ref{eq:state7}),
$m(t_1) =\cdots=m(t_n) = 1$.  So, $m(t_1)+\cdots+m(t_n) = n$.
Using Eq.~(\ref{eq:mge2}), we expand all disjuncts into
sums of states on individual variables:
\begin{eqnarray}
       m(a_{1,1})+\cdots+m(a_{n,p_n}) = n.   \nonumber
\end{eqnarray}
Now, using condition 2 of the MGE definition (``each variable occurs the
same number of times on each side of the equality''), we rearrange this
sum in the form
\begin{eqnarray}
       m(b_{1,1})+\cdots+m(b_{n,q_n}) = n.  \nonumber
\end{eqnarray}
Using Eq.~(\ref{eq:mge2}) again, we collapse the variables back
into the disjunctions
on the right-hand side of the equation:
\begin{eqnarray}
       m(u_1)+\cdots+m(u_n) = n    \nonumber
\end{eqnarray}
Using Eq.~(\ref{eq:state6}), $m(u_1) =\cdots=m(u_n) = 1$.

To summarize:  we have proved so far that for any $u_i$ and any state $m$,
\begin{eqnarray}
      m(t_1 \cap\cdots\cap t_n) = 1\quad\Rightarrow\quad m(u_i) = 1
             \label{eq:asterisk}
\end{eqnarray}
Since $\mathcal L$ admits a strong set of states, there exists a state $m$
such that
\begin{eqnarray}
     (m(t_1 \cap\cdots\cap t_n) = 1  \ \Rightarrow\    m(u_i) = 1)
    \quad\Rightarrow\quad  t_1 \cap\cdots\cap t_n \le u_i.
 \nonumber
\end{eqnarray}
Detaching Eq.~(\ref{eq:asterisk}), we have
\begin{eqnarray}
     t_1 \cap\cdots\cap t_n \le u_i.  \nonumber
\end{eqnarray}
Combining for all $i$, we have
\begin{eqnarray}
     t_1 \cap\cdots\cap t_n \le u_1 \cap\cdots\cap u_n.  \nonumber
\end{eqnarray}
By symmetry
\begin{eqnarray}
     u_1 \cap\cdots\cap u_n \le t_1 \cap\cdots\cap t_n,   \nonumber
\end{eqnarray}
so the Mayet-Godowski equation holds.
\end{proof}

In order to represent MGEs efficiently, we introduce a
special notation for them.  Consider the following MGE (which will
be of interest to us later):
\begin{eqnarray}
&&a\perp b\ \&\ a\perp c\ \&\ b\perp c\ \&\ d\perp e\ \&\ f\perp g\ \&\
h\perp j\ \&\ g\perp b\ \&\ \nonumber\\
&& \qquad \qquad e\perp c\ \&\ j\perp a\ \&\ h\perp f\ \&\
h\perp d\ \&\ f\perp d \ \Rightarrow\ \nonumber\\
&& \qquad (a\cup b\cup c)\cap(d\cup
e)\cap(f\cup g)\cap(h\cup j)=\nonumber\\
&& \qquad \qquad (g\cup b)\cap(e\cup c)\cap(j\cup
a)\cap(h\cup f\cup d).\label{eq:newst1}
\end{eqnarray}
Following the proof of Theorem~\ref{th:mge}, this equation arises from
the following equality involving states:
\begin{eqnarray}
&& m(a\cup b\cup c)+m(d\cup e)+m(f\cup g)+m(h\cup j) = \nonumber\\
&& \qquad m(g\cup b)+m(e\cup c)+m(j\cup a)+m(h\cup f\cup d).
 \label{eq:newst1b}
\end{eqnarray}
A {\em condensed state equation} is an abbreviated representation of this
equality, where\-in we represent
join by juxtaposition and remove all mentions of the state function, leaving
only its arguments.  Thus the condensed state equation representing
Eq.~(\ref{eq:newst1b}), and thus Eq.~(\ref{eq:newst1}), is:
\begin{eqnarray}
abc+de+fg+hj&=&gb+ec+ja+hfd.
 \label{eq:newst1c}
\end{eqnarray}

Another example of an MGE shows that repeated or {\em degenerate} terms
may be needed in the condensed state equation in order to balance the
number of variable occurrences on each side:
\begin{eqnarray}
&& ab+cde+fg+fg+hjk+lk+mn+pe\ = \nonumber\\
&& \qquad\qquad gk+gk+db+fe+fe+nlc+pja+mh\label{eq:st2new}
\end{eqnarray}

\begin{theorem}\label{th:mgo-in-ngo}
The family of all Mayet-Godowski equations
includes, in particular, the Godowski equations
{\rm [Eqs.~(\ref{eq:godow3o}),
(\ref{eq:godow4o}),\ldots]}; in other words, the class {\rm MGO}
is included in $n${\rm GO} for all $n$.
\end{theorem}
\begin{proof}
We will give the proof for 3-Go.  The proofs for $n>3$ are analogous.

To represent 3-Go,
\begin{eqnarray}
(a\to  b)\cap(b\to  c)\cap(c\to  a)&=&
  (c\to  b)\cap(b\to  a)\cap(a\to  c),\label{3go}
\end{eqnarray}
we express it in the form shown by Theorem~\ref{th:go2n}:
\begin{eqnarray}
\lefteqn{a\perp d\perp b\perp e\perp c\perp f\perp a
    \qquad\Rightarrow} & & \nonumber \\
 & & (a\cup d)\cap (b\cup e)\cap(c\cup f)\ \le\ d\cup b.
\label{eq:3goa}
\end{eqnarray}
By symmetry, this is equivalent to the MGE
\begin{eqnarray}
\lefteqn{a\perp d\perp b\perp e\perp c\perp f\perp a
    \qquad\Rightarrow} & & \nonumber \\
& & (a\cup d)\cap (b\cup e)\cap(c\cup f)\ =\
(d\cup b)\cap (e\cup c)\cap(f\cup a),
\label{eq:3gob}
\end{eqnarray}
whose condensed state equation is
\begin{eqnarray}
ad+be+cf&=&db+ec+fa\label{eq:3goc}
\end{eqnarray}
\end{proof}

While every MGE holds in a Hilbert lattice, many of them are derivable
from the equations $n$-Go and others trivially hold in all OMLs.  
An MGE is ``interesting'' if it does not hold in all $n$GOs.  To find such
MGEs, we seek OMLs that are $n$GOs for all $n$ but have no strong set of
states.  Once we find such an OML, it is possible to deduce an MGE that
it will violate.

The search for such OMLs was done with the assistance of several
computer programs written by Brendan McKay and Norman Megill.  An
isomorph-free, exhaustive list of finite OMLs with certain
characteristics was generated.  The ones admitting no strong set of
states were identified (by using the simplex linear programming
algorithm to show that the constraints imposed by a strong set of states
resulted in an infeasible solution).  Among these, the ones violating
some $n$-Go were discarded, leaving only the OMLs of interest.  (To
identify an OML of interest, a special dynamic programming algorithm,
described in \cite{mp-gen-godowski06-arXiv}, was used.  This algorithm was
crucial for the results in this section, providing a proof that the OML
``definitely'' violated no $n$-Go for all $n$ less than infinity, rather
than just ``probably'' as would be obtained by testing up to some large
$n$ with a standard lattice-checking program.)  Finally, an MGE was
``read off'' of the OML, using a variation of a technique described by
Mayet \cite{mayet86} for producing an equation that is violated by a
lattice admitting no strong set of states.

Fig.~\ref{fig:st1new} shows examples of such OMLs found by these
programs.  Eq.~(\ref{eq:newst1}) was deduced from OML MG1 in the figure,
and it provides the answer (Theorem~\ref{th:mgo-lt-ngo} below) to the
problem posed at the beginning of this section.  In order to show how we
constructed Eq.~(\ref{eq:newst1}), we will show the details of the proof
that OML MG1 admits no strong set of states.  That proof will provide us
with an algorithm for stating an equation that fails in OML MG1 but
holds in all OMLs admitting a strong set of states.

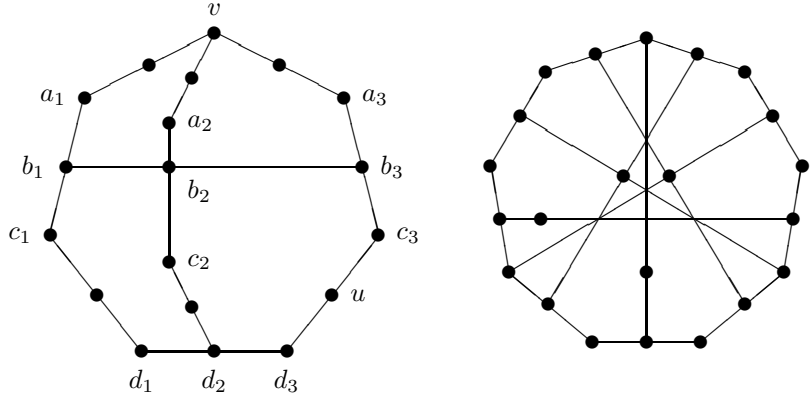
\begin{figure}[htbp]\centering
  \setlength{\unitlength}{1pt}
  \begin{picture}(260,150)(-10,-10)

    \put(-15,0) { 
      \begin{picture}(124,120)(0,0) 
        \put(110.9,96.2){\line(-2,1){49.3}}
        \put(110.9,96.2){\line(1,-4){13}}
        \put(89,0){\line(4,5){35}}
        \put(89,0){\line(-1,0){54.8}}
        \put(34.2,0){\line(-4,5){35}}
        \put(12.2,96.2){\line(-1,-4){13}}
        \put(12.2,96.2){\line(2,1){49.3}}
        \put(5.5,70){\line(1,0){112}}
        \put(61.5,120.5){\line(-1,-2){17}}
        \put(44.5,86.5){\line(0,-1){52.5}}
        \put(61.5,0){\line(-1,2){17}}
        \put(61.5,120.5){\circle*{5}}
            \put(61.5,127.5){\makebox(0,0)[b]{$v$}}
        \put(110.5,96.2){\circle*{5}}
            \put(117.5,96.2){\makebox(0,0)[l]{$a_3$}}
        \put(123.5,44){\circle*{5}}
            \put(130.5,44){\makebox(0,0)[l]{$c_3$}}
        \put(89,0){\circle*{5}}
            \put(89,-7){\makebox(0,0)[t]{$d_3$}}
        \put(34,0){\circle*{5}}
            \put(34,-7){\makebox(0,0)[t]{$d_1$}}
        \put(-0.5,44){\circle*{5}}
            \put(-7.5,44){\makebox(0,0)[r]{$c_1$}}
        \put(12.5,96.2){\circle*{5}}
            \put(5.5,96.2){\makebox(0,0)[r]{$a_1$}}
        \put(86.2,108.4){\circle*{5}}
        \put(117.5,70){\circle*{5}}
            \put(124.5,70){\makebox(0,0)[l]{$b_3$}}
        \put(106.0,21.4){\circle*{5}}
            \put(113.0,21.4){\makebox(0,0)[l]{$u$}}
        \put(61.5,0){\circle*{5}}
            \put(61.5,-7){\makebox(0,0)[t]{$d_2$}}
        \put(17.1,21.4){\circle*{5}}
        \put(5.5,70){\circle*{5}}
            \put(-2.5,70){\makebox(0,0)[r]{$b_1$}}
        \put(36.9,108.4){\circle*{5}}
        \put(44.5,86.5){\circle*{5}}
            \put(51.5,86.5){\makebox(0,0)[l]{$a_2$}}
        \put(53,103.5){\circle*{5}}
        \put(53,17){\circle*{5}}
        \put(44.5,34){\circle*{5}}
            \put(51.5,34){\makebox(0,0)[l]{$c_2$}}
        \put(44.5,70){\circle*{5}}
            \put(51.5,65){\makebox(0,0)[lt]{$b_2$}}
      \end{picture}
    } 

    \put(150,0) { 
      \begin{picture}(124,110)(0,0) 
        \put(39.48,3.62){\line(1,0){41.04}}
        \put(39.48,3.62){\line(-6,5){31.44}}
        \put(80.52,3.62){\line(6,5){31.44}}
        \put(8.04,30){\line(-1,6){6.75}}
        \put(111.96,30){\line(1,6){6.75}}
        \put(0.91,70.42){\line(3,5){21.2}}
        \put(119.09,70.42){\line(-3,5){21.2}}
        \put(22.13,105.96){\line(3,1){38}}
        \put(97.87,105.96){\line(-3,1){38}}
        \put(60,3.62){\line(0,1){114.98}}
        \put(40.715,112.3){\line(3,-5){56.7}}
        \put(79.285,112.3){\line(-3,-5){56.7}}
        \put(8.04,30){\line(5,3){99.5}}
        \put(111.96,30){\line(-5,3){99.5}}
        \put(4.475,50.21){\line(1,0){111.05}}

        \put(80.52,3.62){\circle*{5}}
        \put(39.48,3.62){\circle*{5}}
        \put(111.96,30){\circle*{5}}
        \put(8.04,30){\circle*{5}}
        \put(119.09,70.42){\circle*{5}}
        \put(0.91,70.42){\circle*{5}}
        \put(97.2,105.8){\circle*{5}}
        \put(21.8,105.8){\circle*{5}}
        \put(60,118.6){\circle*{5}}
        \put(60,3.62){\circle*{5}}
        \put(97.24,17.81){\circle*{5}}
        \put(22.76,17.81){\circle*{5}}
        \put(115.525,50.21){\circle*{5}}
        \put(4.475,50.21){\circle*{5}}
        \put(107.83,89.19){\circle*{5}}
        \put(12.17,89.19){\circle*{5}}
        \put(79.285,112.5){\circle*{5}}
        \put(40.715,112.5){\circle*{5}}
        \put(20,50.21){\circle*{5}}
        \put(60,30.21){\circle*{5}}
        \put(51.3,66.6){\circle*{5}}
        \put(68.7,66.6){\circle*{5}}

      \end{picture}
    } 

  \end{picture}
  \caption{OMLs that admit no strong sets
of states but which are $n$GOs for all $n$.  (a) OML MG1; (b) OML MG5s.
\label{fig:st1new}}
\end{figure}

\begin{theorem}\label{th:mg1-ns}
The {\rm OML} {\rm MG1} does not admit a strong set of states.
\end{theorem}
\begin{proof}
Referring to Fig.~\ref{fig:st1new}, suppose that $m$ is a state such
that $m(v)=1$.  Since the state values of the atoms in a block sum to
$1$, $m(a_1)=m(a_2)=m(a_3)=0$.  Thus $m(b_1)+m(c_1)=
m(b_2)+m(c_2)=m(b_3)+m(c_3)=1$.  Since $m(b_1)+m(b_2)+m(b_3)\le 1$, it
follows that $m(c_1)+m(c_2)+m(c_3)\ge 2$.  Since $m(d_1)+ m(d_2)+m(d_3)=1$,
we have $[m(c_1)+m(d_1)]+[m(c_2)+m(d_2)]+[m(c_3)+m(d_3)]\ge 3$.  Since
$m(c_1)+m(d_1)\le 1$, $m(c_2)+m(d_2)\le 1$, and $m(c_3)+m(d_3)\le 1$, we
must have $m(c_3)+m(d_3)=1$.  Hence $m(u)$=0, since $u$ is on the same
block as $c_3$ and $d_3$.  So, $m(u')=1$.  To summarize, we have shown
that for any $m$, $m(v)=1$ implies $m(u')=1$.  If MG1 admitted a strong
set of states, we would conclude that $v\le u'$, which is a
contradiction since $v$ and $u'$ are incomparable.
\end{proof}

In the above proof, we made use of several specific conditions
that hold for the atoms and blocks in that OML.  That proof was
actually carefully constructed so as to minimize the need for these
conditions.  For example, we used $m(b_1)+m(b_2)+m(b_3)\le 1$ even
though the stronger $m(b_1)+m(b_2)+m(b_3)= 1$ holds, because the
strength of the latter was not required.  The complete set of
such conditions that the proof used are the following facts:
\begin{itemize}
\item $v \perp a_i$, $i=1,2,3$;
\item $d_i \perp c_i$, $i=1,2$;
\item The atoms in each of the triples $\{a_i,b_i,c_i\}$ ($i=1,2,3$),
and $\{d_1,d_2,d_3\}$
are mutually orthogonal and their disjunction is $1$
(i.e. the sum of their state values is 1).
\item The atoms in each of the triples
$\{b_1,b_2,b_3\}$ and $\{c_3,u,d_3\}$
are mutually orthogonal and the sum of their state values is $\le 1$
(the sum is actually equal to $1$, but we used only $\le 1$ for
the proof).
\end{itemize}
If the elements of any OML $\mathcal{L}$ satisfy these
facts, then
we can prove (with a proof essentially identical to that of Theorem
\ref{th:mg1-ns}, using the above facts as hypotheses in place of the
atom and block conditions in OML MG1) that
for any state $m$ on $\mathcal{L}$,
$m(v)=1$ implies $m(u')=1$.  Then, if $\mathcal{L}$ admits
a strong set of states, we also have $v\le u'$.

We can construct an equation that expresses this result as
follows.  We use the orthogonality conditions
from the above list of fact as hypotheses, and we incorporate
each ``disjunction is $1$'' condition as
a conjunct on the left-hand side.
We will denote the set of all orthogonality conditions
in the above list of facts by $\Omega$.
We can ignore the conditions ``the sum of their state values is $\le 1$''
from the above list of facts, because that happens automatically
due to the mutual orthogonality of those elements.
This procedure then leads to the equation,
\begin{eqnarray}
\Omega \ \Rightarrow \ &
v\cap (a_1\cup b_1\cup c_1)\cap(a_2\cup b_2\cup c_2)
\cap(a_3\cup b_3\cup c_3)\cap \nonumber\\
&(d_1\cup d_2\cup d_3)\le u'
\end{eqnarray}
This equation holds in all OMLs with a strong set of states
but fails in lattice MG1.

The condensed state equation
Eq.~(\ref{eq:newst1c}) was obtained using the following
mechanical procedure.  We consider only variables corresponding to
the atoms used by the proof (i.e. the labeled atoms in
Fig.~\ref{fig:st1new}) and only the blocks whose orthogonality
conditions were used as hypotheses for the proof.  We ignore all
variables whose state value is shown to be equal to $1$ or $0$ by the
proof, and we ignore all blocks in which only one variable remains as a
result.  For the left-hand side, we consider all the remaining blocks
that have ``disjunction is $1$'' in the assumptions listed above.  We
juxtapose the (unignored) variables in each block to become a term, and
we connect the terms with $+$.  For the right-hand side, we do the same
for the remaining blocks that do not have ``disjunction is $1$'' in the
assumptions listed above.  Thus we obtain:
\begin{eqnarray}
  b_1 c_1 + b_2 c_2 + b_3 c_3 + d_1 d_2 d_3 &=&
    c_1 d_1 + c_2 d_2 + c_3 d_3 + b_1 b_2 b_3
\end{eqnarray}
After renaming variables and rearranging terms, this is
Eq.~(\ref{eq:newst1c}), which corresponds to the
MGE Eq.~(\ref{eq:newst1}) and which can be verified to
fail in lattice MG1.

This mechanical procedure is simple and practical to automate---the simplex
algorithm used to find states lets us determine which blocks must have a
disjunction equal to 1---but it is not guaranteed to be successful in
all cases:  in particular, it will not work when the condensed state
equation has degenerate terms, as in Eq.~(\ref{eq:st2new}) above.
However, such cases are easily identified by counting the variable
occurrences on each side, and we can add duplicate terms to make the
counts balance in the case of a degeneracy.  This balancing ensures that
the corresponding equation is an MGE and therefore holds in all Hilbert
lattices.

Having constructed Eq.~(\ref{eq:newst1}), which holds in all Hilbert
lattices but fails in lattice MG1, we now state the main result
of this section.

\begin{theorem}\label{th:mgo-lt-ngo}
The class {\rm MGO} is properly included in all $n${\rm GO}s, i.e.,
not all {\rm MGE} equations can be deduced from the equations
$n$-{\rm Go}.
\end{theorem}
\begin{proof}
We have already shown that MGO is included in all $n$GOs
(Theorem~\ref{th:mgo-in-ngo}).  Furthermore,
OML MG1 is an $n$GO for all $n$, but is not an MGO.  Specifically,
it can be shown that the equations $n$-Go hold in OML MG1 for all $n$,
\cite{mp-gen-godowski06-arXiv}
whereas the MGE Eq.~(\ref{eq:newst1}) fails in OML MG1.
This shows the inclusion is proper.
\end{proof}

In particular, Eq.~(\ref{eq:newst1}) therefore provides an an example of
a new Hilbert lattice equation that is independent from all Godowski
equations.

Having 9 variables and 12 hypotheses, Eq.~(\ref{eq:newst1}) can be
somewhat awkward to work with directly.  It is possible to derive from
it a simpler equation through the use of substitutions that Mayet calls
{\em generators}.  If, in Eq.~(\ref{eq:newst1}), we substitute
(simultaneously) $c'$ for
$a$, $c\cap b$ for $b$, $(c\to  b)'$ for $c$, $(a\to  b)'$ for $d$,
$(c\to  b)\cap(a\to  b)$ for $e$, $b\cap a$ for $f$, $b'$ for $g$,
$a'$ for $h$, and $a\cap c$ for $j$, all of the hypotheses are satisfied
(in any OML) and the conclusion evaluates to:
\begin{eqnarray}
((a\to  b)\to (c\to  b))\cap(a\to  c)\cap(b\to  a)&\le& c\to  a
\label{eq:newst1d}
\end{eqnarray}
where we also dropped all but one conjunct on the right-hand-side.
While such a procedure can sometimes weaken an MGE, it can be verified
that Eq.~(\ref{eq:newst1d}) fails in OML MG1 of
Fig.~\ref{fig:st1new} as desired, thus providing us with a Hilbert
lattice equation that is convenient to work with but is still
independent from all Godowski equations.  For example,
Eq.~(\ref{eq:newst1d}) can be used in place of
Eq.~(\ref{eq:newst1}) to provide a simpler proof of
Theorem~\ref{th:mgo-lt-ngo}.

Eq.~(\ref{eq:st2new}) was deduced from the OML MG5s in
Fig.~\ref{fig:st1new}, and it provides us with another new Hilbert
lattice equation that is independent from all $n$-Gos.  A comparison to
OML G5s in Fig.~\ref{fig:oag6} illustrates how the addition of an atom
can affect the behaviour of a lattice.

\begin{figure}[hbt]
\includegraphics[width=0.9999\textwidth]{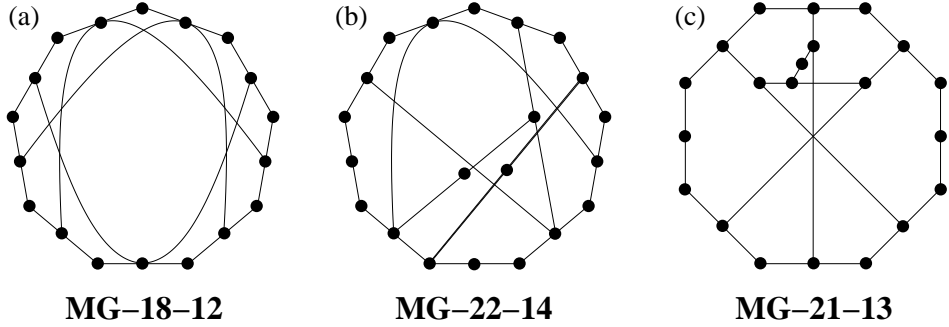}
\caption{OMLs that admit no strong sets of
states but are $n$GOs for all $n$.}
\label{fig:st3_4_19}
\end{figure}

The OMLs of Fig.~\ref{fig:st3_4_19}$\>$(a), (b), and (c) provide
further examples that admit no strong sets of
states but are $n$GOs for all $n$.  The following MGEs
(represented with condensed state equations) can be deduced
from them, respectively:
\begin{eqnarray}
abc+de+fg+hj+kl&=&eb+dh+faj+lc+kg \label{eq:mgeq1}\\
ab+cd+ef+ghj+kl+kl&=&kd+bl+jl+fk+ha+gec \label{eq:mgeq2}\\
abc+def+gh+jk+lmn+pqr&=&fn+rc+dkb+gma+qeh+plj.\qquad\label{eq:mgeq3}
\end{eqnarray}
Using generators, the following examples of simpler Hilbert lattice
equations can be derived from these MGEs, again respectively:
\begin{eqnarray}
&((d\to (a\to b))\cap((a\to c)\to d)\cap(b\to c)
  \cap(c\to a)\quad\le\quad b\to a  \\
&(d\to (c\cap(a\to b))\cap((b\to a)\to d)\cap(c\to a)\cap(b\to d)
  \quad\le\quad a\to c  \\
&((d\to a)\to (b\to c)')\cap((c\to d)\to (a\to b)')\cap
((b\to a)'\to (d\to c))\cap \nonumber \\
&((a\to d)'\to (c\to b))
  \quad\le\quad (d\to c)\to (b\to a)'
\end{eqnarray}
Each of these simpler equations, while possibly weaker than the MGEs
they were derived from, still fail in their corresponding OMLs, thus
providing us with additional new Hilbert lattice equations that are
independent from all $n$GOs.

While the complete picture of interdependence of the three lattice
families we have presented ($n$OA, $n$GO, and MGO) is not fully
understood, some results can be established.  We have
already shown that every MGO is an $n$GO for all $n$, and moreover that
the inclusion is proper (Theorem~\ref{th:mgo-lt-ngo} ). We can also
prove the following:

\begin{theorem}\label{th:mgo-not-3oa}
There are {\rm MGO}s {\rm (}and therefore $n${\rm GO}s{\rm )} that
are not {\rm 3OA}s and thus not $n${\rm OA}s for any $n$.
\end{theorem}
\begin{proof}
The OML 13-7-OMLp-oa3f of Fig.~\ref{fig:oa4f-5f} has a strong set
of states and thus is an MGO.  However, it violates the 3OA law.
\end{proof}

\begin{theorem}\label{th:oa-not-3go}
There are $n${\rm OA}s for $n=3,4,5,6$ that are not {\rm 3GO}s and thus
not $n${\rm GO}s for any $n$ nor {\rm MGO}s.
\end{theorem}
\begin{proof}
The OML of Fig.~\ref{fig:oa-not-3go} is a 6OA that is not a 3GO.
\end{proof}

\begin{figure}[hbt]
\begin{center}
\includegraphics[width=0.8\textwidth]{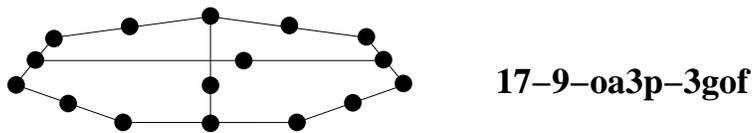}
\end{center}
\caption{An OML that is $n$OA for $n=3,4,5,6$ but which is neither
$n$GO for any $n$ nor MGO.}
\label{fig:oa-not-3go}
\end{figure}

Whether Theorem~\ref{th:oa-not-3go} holds for all $n$OAs remains an open
problem.  However, our observation is that the smallest OMLs in which
the $n$OA law passes but the $(n+1)$OA law fails grow in size with
increasing $n$, as indicated by the OMLs used to prove Theorem
\ref{th:stronger}.  Compared to them, the OML of
Fig.~\ref{fig:oa-not-3go} is ``small,'' leading us to conjecture that it
is an $n$OA for all $n$.  If this conjecture is true, it would show that
no $n$-Go equation can be derived (in an OML) from the $n$OA laws.

Figure \ref{fig:hl-chart} summarizes the relationships among the known
families of Hilbert lattice equational varieties.  In particular, it is
unknown whether MGO, any $n$OA, or any $n$GO contains the class of
modular OLs, with the single exception of 3OA.  In other words, the modular
law implies the 3OA law in an OL, but it is unknown if it implies any
of the other stronger-than-OML equations discussed in this chapter.

\begin{figure}[hbt]\centering
  \setlength{\unitlength}{0.9pt}
\begin{picture}(350,250)(45,0)
\multiput(45,165)(10,0){34}{\line(1,0){5}}

\put(190,216){\framebox(60,20){Boolean}}
\put(220,195){\vector(0,1){21}}
\put(175,174){\framebox(90,20){Modular OL}}
\put(380,172){\makebox(0,0)[r]{non-Hilbert lattices}}

\put(380,160){\makebox(0,0)[r]{Hilbert lattices}}
\put(330,147){\vector(-1,1){80}}
\put(300,156){\makebox(0,0)[l]{?}}
\multiput(315,137)(-3,3){15}{\line(1,0){0.8}}
     \put(275,177){\vector(-1,1){10}}
\put(140,156){\makebox(0,0)[r]{?}}

\put(155,157){\makebox(0,0)[l]{$n=3$}}
\put(155,148){\makebox(0,0)[l]{proved}}

\multiput(125,137)(3,3){15}{\line(1,0){0.8}}
     \put(165,177){\vector(1,1){10}}
\put(110,147){\vector(1,1){80}}
  \put(65,126){\framebox(60,20){$n$OA}}
  \put(315,126){\framebox(60,20){MGO}}

\put(345,105){\vector(0,1){21}}
\put(90,105){\vector(0,1){21}}
  \put(65,84){\framebox(60,20){OA}}
  \put(315,84){\framebox(60,20){$n$GO}}
\put(250,52){\vector(3,2){65}}
\put(220,63){\vector(0,1){111}}
\put(190,52){\vector(-3,2){65}}

\put(330,126){\vector(0,-1){21}}
\put(326,112){\line(1,1){10}}
\put(310,124){\makebox(0,0)[r]{Megill/Pavi\v ci\'c}}
\put(310,112){\makebox(0,0)[r]{(2006)}}
\put(105,126){\vector(0,-1){21}}
\put(101,112){\line(1,1){10}}
\put(130,124){\makebox(0,0)[l]{Megill/Pavi\v ci\'c}}
\put(130,112){\makebox(0,0)[l]{(2000)}}

\put(295,15){\vector(1,0){15}}
\put(315,15){\makebox(0,0)[l]{means $\supset$}}
\put(295,3){\vector(1,0){15}}
\put(297,8){\line(1,-1){10}}

\put(315,3){\makebox(0,0)[l]{means $\nsupseteq$}}

  \put(190,42){\framebox(60,20){OML}}
\put(220,21){\vector(0,1){21}}
  \put(190,0){\framebox(60,20){OL}}
\end{picture}
\caption{Known relationships among equational varieties of
Hilbert lattices.}
\label{fig:hl-chart}
\end{figure}
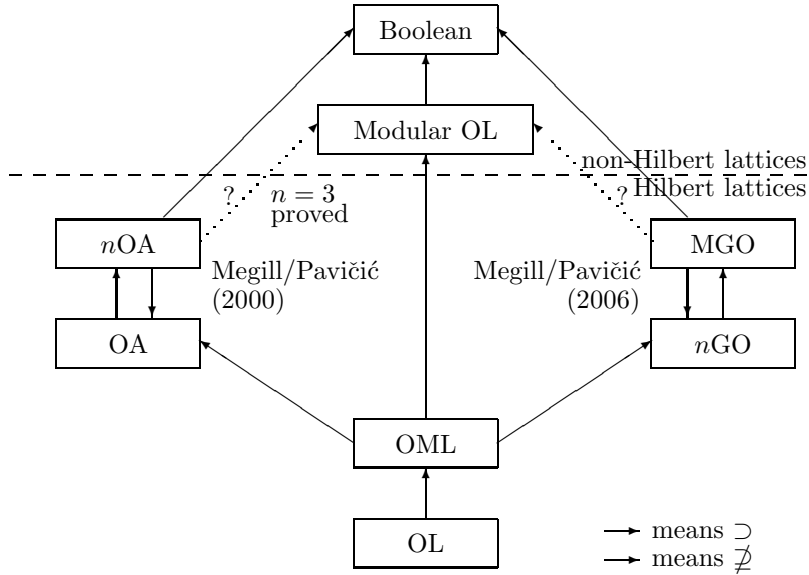

\section{State Vectors: Mayet's E-Equations}
\label{sec:mayet}

In the three previous sections we have presented two
apparently very different ways of generating Hilbert
lattice equations. The first one was algebraic, utilizing
an algebraic formulation of a geometric property possessed
by any Hilbert space. The second one was based on the
the properties of states (probability measures) one can
define on any Hilbert space. Theorem \ref{th:repr} in
Section \ref{sec:hl} offers us a property of a third kind
which any Hilbert space possesses and which can generate
a class of Hilbert lattice equations, and this is that
each Hilbert space is defined over a particular field
$\mathcal K$.

The application to quantum theory uses the Hilbert
spaces defined over real, $\mathcal R$,  complex, $\mathcal C$,
or quaternion (quasi), $\mathcal Q$, fields. For these fields,
in 2006, Ren{\'e} Mayet \cite{mayet06} (see also 
\cite{mayet06-hql2}) used a technique
similar to the one used for generating MGEs we presented in
Sec.~\ref{sec:mge}, to arrive at a new class of E-equations
we will present in this section. There are other fields over
infinite-dimensional Hilbert spaces, for example a
non-archimedean Keller field.~\cite{keller,gross,soler},
so, to get only the aforementioned three fields for an
infinite-dimensional Hilbert space, we have to assume 
infinite orthonormality and invoke the theorem of Maria Pia
Sol{\`e}r \cite{soler} (Theorem \ref{th:sol}). If we do not
have an infinite orthonormal series of vectors, then,
for an arbitrary vector $a\in\mathcal H$, a vector
$b\in {\mathcal K}a$, satisfying
$(b,b)=1_{\mathcal K}$, where $(\,,)$ is the inner product in 
$\mathcal H$, might not exist.
If we have an  orthonormal series of vectors, we will
always have vectors satisfying the condition
$(b,b)=1_{\mathcal K}$, and this enables us to
introduce {\em Hilbert-space-valued states}\/$\,$\footnote{One
could also name them {\em vector states} because they map
elements of a Hilbert lattice to state vectors of the Hilbert
space, but we decided to keep to the name introduced by
Mayet.\ \cite{mayet06}}
as follows.

\begin{definition}\label{def:rh-state} A {\em real
Hilbert-space-valued state}---we call it an $\mathcal{RH}$
{\em state}---on an orthomodular lattice $\mathcal L$
is a function $s:{\mathcal L}\longrightarrow \mathcal{RH}$,
where $\mathcal{RH}$ is a Hilbert space defined over a real
field, such that

\begin{itemize}
\item[] $||s(1_{\mathcal L})||=1$,
where $s(a)\in {\mathcal RH}$ is a state vector, 
$||s(a)||=\sqrt{(s(a),s(a))}$ is the Hilbert space norm, and
$a\in {\mathcal L}$; in this section we will not use the 
Dirac notation $|s\rangle$ for the state vector $s$, nor $\langle s|t\rangle$ 
for the inner product $(s,t)$; 
\item[] $(\forall a,b\in {\mathcal L})\left[\>a\perp b\
\Rightarrow\ s(a\cup b)=s(a)+
s(b)\>\right]$,
where $a\perp b$ means $a\le b'$;
\item[] $(\forall a,b\in {\mathcal L})\left[\>a\perp b\
\Rightarrow\ s(a)\perp s(b)\>\right]$,
where $\>s(a)\perp s(b)\>$ means
the inner product $(s(a),s(b))=0$.
\end{itemize}

Now, we select those Hilbert lattices in which we implement
Definition \ref{def:rh-state} by the following definition.
\begin{definition}\label{def:qhl} A {\em
quantum\/$\,$\footnote{Mayet \cite{mayet06} calls this lattice
{\em classical Hilbert lattice} but
since the real and complex fields as well as the quaternion
skew filed over which the corresponding Hilbert space is
defined are characteristic of its application in quantum
mechanics we prefer to call the lattice quantum.} Hilbert
lattice}, $\mathcal{QHL}$, is a Hilbert lattice orthoiso\-morphic
to the set of closed subspaces of the Hilbert space defined
over either a real field, or a complex field, or a quaternion
skew field.
\end{definition}

In 1998 Ren{\'e} Mayet \cite{mayet98} showed that each of
the three quantum Hilbert lattices can be given a rigorous
definition so that the three classes of them are proper
subclasses of the class of all Hilbert lattices.
However, as with equations in the previous
sections, we shall use only some properties related to
states defined on a $\mathcal{QHL}$, in particular pairwise
orthogonality of its elements---corresponding to
pairwise orthogonality of vectors in the corresponding Hilbert
space---to arrive at new equations.

We also define a {\em complex} and a
{\em quaternion Hilbert-space-valued state}, called a $\mathcal{CH}$
{\em state} and a $\mathcal{QH}$
{\em state}, by mapping $s$ to
$\mathcal{CH}$ or $\mathcal{QH}$, i.e. a Hilbert
space defined over a complex or quaternion field respectively.
\end{definition}

This definition differs from Definition \ref{def:state}
in a crucial point, in that the
state does not map the elements of the lattice
to the real interval $[0,1]$ but instead to the real
Hilbert space $\mathcal{RH}$. In particular, the property
$a\perp b\ \Rightarrow\ s(a)\perp s(b)$ is a
a restrictive requirement that allows us to define a strong
set of $\mathcal{RH}$ states on a $\mathcal{QHL}$ but
not on OMLs in general---even those admitting strong
sets of real-valued states---nor even on all Hilbert
lattices.

The conditions of Lemma~\ref{lem:state} hold when
we replace a real state value $m(a)$ with the square
of the norm of the $\mathcal{RH}$ state value $s(a)$.
For example, Eq.~(\ref{eq:state3}) becomes
\begin{eqnarray}
||s(a)||^2 + ||s(a')||^2 = 1,
\end{eqnarray}
and so on.  In addition, we can prove the following special
properties that hold for $\mathcal{RH}$ states:

\begin{lemma}\label{lem:rhstate}
The following properties hold for any $\mathcal{RH}$ state $s$:
\begin{eqnarray}
& s(0)=0 \label{eq:rhstate1}\\
& s(a)+s(a')=s(1)  \label{eq:rhstate2}\\
& ||s(a)||=1 \qquad \Leftrightarrow \qquad s(a)=s(1)  \label{eq:rhstate3}\\
& ||s(a)||=0 \qquad \Leftrightarrow \qquad s(a)=s(0)  \label{eq:rhstate4}\\
& s(a)\perp s(1) \qquad \Leftrightarrow \qquad s(a)=0  \label{eq:rhstate5}\\
& a\perp b \qquad \Rightarrow \qquad ||s(a\cup b)||^2 =
  ||s(a)||^2 + ||s(b)||^2  \label{eq:rhstate6}\\
& a\le b \qquad \Rightarrow \qquad ||s(a)||\le ||s(b)||
      \label{eq:rhstate7}\\
& a\le b \qquad \& \qquad ||s(a)||=1 \qquad \Rightarrow \qquad ||s(b)||=1
    \label{eq:rhstate8}\\
& a_i \perp a_j (1\le i<j\le n) \quad
   \& \quad a_1\cup\cdots\cup a_n=1 \qquad \Rightarrow\nonumber\\
& \qquad\qquad\qquad
   s(a_1)+\cdots +s(a_n)=s(1) \label{eq:rhstate9}
\end{eqnarray}
\end{lemma}
\begin{proof}
Some of these conditions are proved in \cite{mayet06}, and the others are
straightforward consequences of them.
\end{proof}
The conditions of
Lemma~\ref{lem:rhstate}, as well as the analogues of Lemma~\ref{lem:state},
also hold for $\mathcal{CH}$ and $\mathcal{QH}$ states.

The following definition of a strong set of $\mathcal{RH}$ states
closely follows Definition \ref{def:strong}, with an essential
difference in the range of the states.

\begin{definition}\label{def:strong-hs}
A nonempty set $S$ of $\mathcal{RH}$ states
$s:{\mathcal L}\longrightarrow {\mathcal{RH}}$ is
called a {\em strong set of $\mathcal{RH}$} states if
\begin{eqnarray}
(\forall a,b\in{\mathcal L})(\exists s \in S)((||s(a)||=1\ \Rightarrow
\ ||s(b)||=1)\ \Rightarrow\ a\le b)\,.\quad
\label{eq:st-rhs}
\end{eqnarray}
In an analogous manner, we define
a {\em strong set of $\mathcal{CH}$} states and
a {\em strong set of $\mathcal{QH}$} states.
\end{definition}

The following version of Theorem \ref{hilb-strong-s}
holds.~\cite{mayet06}

\begin{theorem}\label{rh-strong-s} Any quantum Hilbert
lattice admits a strong set of $\mathcal{RH}$ states.
\end{theorem}
\begin{proof}
Let $\mathcal{L}$ be a $\mathcal{QHL}$.
For each $u$ in the proof of Theorem~\ref{hilb-strong-s}, define
$s(a)=P_a(u)$.  We thus have $s:\mathcal{L}\longrightarrow
{\mathcal{H}}$, where $\mathcal{H}$ is $\mathcal{RH}$, $\mathcal{CH}$,
or $\mathcal{QH}$ according to the field underlying $\mathcal{L}$.
A theorem for projectors tells us that $a\perp b \ \Rightarrow \
P_a(u)\perp P_b(u)$, showing that $s$ satisfies the third condition of
Def.~\ref{def:rh-state}.  The other two conditions are easy to verify,
so $s$ is a $\mathcal{H}$ state.  Observing that
$m(a)=||s(a)||^2$ in the proof of Theorem~\ref{hilb-strong-s}, a nearly
identical proof shows that $\mathcal{L}$ admits a strong set of
$\mathcal{H}$ states.  Since any OML admits a strong set of
$\mathcal{RH}$ states iff it admits a strong set of $\mathcal{CH}$ iff
it admits a strong set of $\mathcal{QH}$ states, \cite{mayet06} we
conclude that any $\mathcal{QHL}$ admits a strong set of $\mathcal{RH}$
states.
\end{proof}

Now, Mayet \cite{mayet06} showed that the lack of
$\mathcal{RH}$ strong states for particular lattices,
for example, the ones given in Figure \ref{fig:en},
gives the equations in the way similar to the one
used by Megill and Pavi\v ci\'c \cite{mp-gen-godowski06-arXiv}.
For certain infinite sequences of equations, Mayet's
method offers the advantage of providing a related
infinite sequence of finite OMLs that violate the
corresponding equation, analogous to the wagon-wheel
series obtained by Godowski and presented in
Section \ref{sec:ge}.

\begin{figure}[hbt]
\begin{center}
\includegraphics[width=0.8\textwidth]{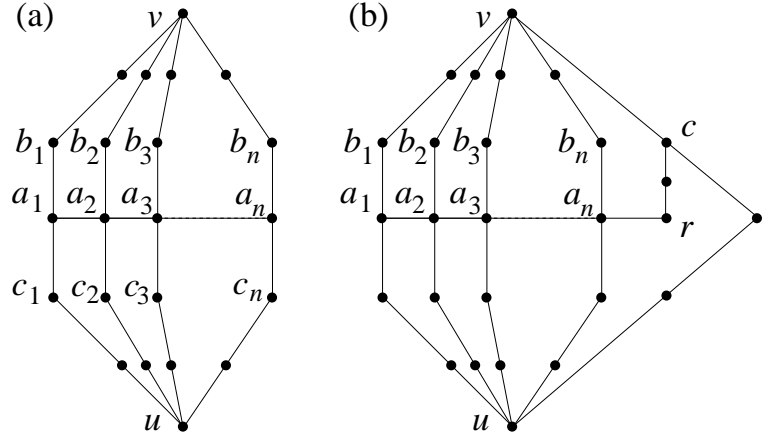}
\end{center}
\caption{Greechie diagrams $\mathcal{L}_n$ in which $E_n$
fail and which serve to generate $E_n$.}
\label{fig:en}
\end{figure}

Let us first denote by $\Omega$ the following
set of orthogonality conditions
among the labeled atoms in Figure \ref{fig:en}$\>$(a):
$\Omega=\{v\perp b_i,\
b_i\perp a_i,\
\ a_i\perp a_j\}$, $i,j=1,\dots,n$.
Next, we define
\begin{eqnarray}
a = a_1 \cup \cdots \cup a_n,\quad
q = (a_1 \cup b1) \cap \cdots \cap (a_n \cup b_n),\quad
b = b_1 \cup\cdots\cup b_n\,.\label{eq:aqb}
\end{eqnarray}

Now we are able to generate the following equations, i.e.,
to prove the following theorem.

\begin{theorem} \label{th:mayet-eqs}
In $\mathcal{L}_i$, $i=1,\dots,n$, $n\ge 3$ given in Figure
\ref{fig:en}$\>$(a),(b) the following equations fail
\begin{eqnarray}
&E_n:\quad \Omega&\quad  \Rightarrow
\quad a \cap q = b
\label{eq:mayet-eqs-en-a}\\
&E'_n:\quad \Omega& \&\quad r\perp a\quad
\Rightarrow \quad q \cap (q\to r')\cap(a\cup r)\le b
\label{eq:mayet-eqs-en-b}
\end{eqnarray}
respectively and they hold in any OML with a strong set of
$\mathcal{RH}$ states.
\end{theorem}

\begin{proof}
We will show details of proof for equations $E_n$.  The
proof for $E'_n$ involves similar ideas, and we refer the reader to
Mayet.\ \cite{mayet06}

First we show that the OML of Figure \ref{fig:en}$\>$(a) does not admit
a strong set of $\mathcal{RH}$ states.  Referring to the atoms labeled
in the figure, suppose that $s$ is a state such that $s(u)=s(1)$.  By
Eq.~(\ref{eq:rhstate9}), the condition $s(u)=s(1)$ implies that the
state value of all other atoms in the blocks that atom $u$ connects to
are $0$, so $s(c_i)=0$ and thus $s(a_i)+s(b_i)=s(1)$, $i=1,\ldots,n$.
Summing these then using $s(a_1)+\cdots +s(a_n)=s(1)$, we obtain
$s(a_1)+\cdots +s(a_n)+s(b_1)+\cdots +s(b_n)=ns(1)=
s(1)+s(b_1)+\cdots +s(b_n)$, or $s(b_1)+\cdots +s(b_n)=(n-1)s(1)$.  The
primary feature that distinguishes real-valued states and $\mathcal{RH}$
states now comes into play:  from the third condition in Definition
\ref{def:rh-state}, $v\perp b_i$ implies $s(v)\perp s(b_i)$ for
$i=1,\ldots,n$.  Thus $s(v)\perp s(b_1)+\cdots +s(b_n)$ i.e.\
$s(v)\perp(n-1)s(1)$.  By Eq.~(\ref{eq:rhstate5}), then, $s(v)=0$ and
$s(v')=s(1)$.  To summarize, we have shown that for any state $s$, if
$||s(u)||=1$ then $||s(v')||=1$ in the OML of Figure
\ref{fig:en}$\>$(a).  If the OML admitted a strong set of $\mathcal{RH}$
states, we would have $u\le v'$, which is not true since those atoms are
incomparable.  This shows the OML does not admit a strong set of
$\mathcal{RH}$ states.

In the above proof, we used the following facts:
\begin{itemize}
\item The labeled atoms Figure \ref{fig:en}$\>$(a)
  that belong to the same block are
  mutually orthogonal;
\item The disjunctions $a_i \cup b_i \cup c_i = 1$ for $i=1,\ldots,n$;
\item The disjunction $a_1 \cup \cdots \cup a_n=a=1$.
\end{itemize}

If the elements of any OML $\mathcal{L}$ satisfies these
facts, then
we can prove (with a proof essentially identical to the one
above, using the above facts as hypotheses in place of the
atom and block constraints in the OML of Figure
\ref{fig:en}$\>$(a)) that
for any state $s$ on $\mathcal{L}$,
$||s(u)||=1$ implies $||s(v')||=1$.  Then, if $\mathcal{L}$ admits
a strong set of $\mathcal{RH}$ states, we also have $u\le v'$.
We write down an equation expressing these conditions as follows:
\begin{eqnarray}
&E^1_n:\quad \Omega^1&\quad  \Rightarrow
\quad u\cap a\cap(a_1\cup b_1\cup c_1)\cap
\cdots\cap(a_n\cup b_n\cup c_n)\le v',\quad
\label{eq:mayet-eqs-e1n-a}
\end{eqnarray}
where $\Omega^1=\Omega\cup \{b_i\perp c_i,
\ a_i\perp c_i,\ u\perp c_i\}$, $i=1,\dots,n$.
The ``disjunction = $1$'' conditions used by the
proof are incorporated
as terms on the left-hand side of the inequality.
With some substitutions and manipulations,
equation $E^1_n$ can be shown OML-equivalent to
$E_n$.\ \cite{mayet06}
\end{proof}

The equations of Theorem~\ref{th:mayet-eqs}, which hold in every
$\mathcal{QHL}$, do not hold in every $\mathcal{HL}$.
Thus they are independent from all of the equations we have
presented in Secs.~\ref{sec:goe}, \ref{sec:ge}, and
\ref{sec:mge}.  In addition, they are independent of the modular law.

\begin{theorem} \label{th:mayet4.2}
For any integer $n\ge 3$, the equation $E_n$ does not hold in every
$\mathcal{HL}$.  In particular, it is not a consequence of any $n${\rm OA}
law, $n${\rm GO} law, {\rm MGE}, or combination of them.  In addition,
it is not a consequence of these even in the presence of the modular law.
\end{theorem}
\begin{proof}
The definition of $\mathcal{HL}$ does not require any special property of
the underlying field of the intended Hilbert space.  However,
Theorem~4.1 in \cite{mayet06} shows that for some fields and some finite
dimensions, $E_n$ fails.  Since the other mentioned equations,
including the modular law, hold at least in every finite-dimensional
$\mathcal{HL}$, $E_n$ is independent from them.
\end{proof}

Mayet has also generalized the direct-sum decomposition method used in
the proof of the $n$OA laws (Theorem~\ref{th:noa}) to result in an
additional series of equations.\ \cite{mayet06} However, so far it
is unknown whether any of them are not consequences of some $n$OA law,
and additional investigation is needed.

\section{Conclusion}
\label{sec:d}

In the previous sections we reviewed the results obtained in
the field of Hilbert space equations. The idea is to use
classes of Hilbert lattice equations for an alternative representation
of Hilbert lattices and Hilbert spaces of arbitrary quan\-tum
systems that might enable a direct introduction of the
states of the systems into quantum computers. More specifically,
we were looking for a way to feed a quantum computer with
algebraic equations of $n$th order underlying an infinite
dimensional Hilbert space description of quantum systems.

Quantum computation, at its present stage, manipulates quantum bits
$\{|0\rangle,|1\rangle\}$ by means of quantum logic gates (unitary
operators), following algorithms for computing particular problems.  In
the Introduction, we presented one such gate, the Mach-Zehnder
interferometer.  Quantum gates are integrated into quantum circuits that
represent quantum computers.  The quantum algebra of such circuits is
the algebra of the finite-dimensional Hilbert space which describes the
states of qubits manipulated by quantum gates. We call it
qubit algebra.

A general quantum algebra underlying a description of general
quantum systems such as atoms and molecules is much more
complicated than qubit algebra because it includes continuous
observables, which require infinite-dimensional Hilbert space.
The algebra is called the Hilbert lattice, and we presented it
in Section \ref{sec:hl}. Its connection to measurement and
the standard Hilbert space formalism is given by the {\em
Gleason theorem}.\ \cite{gleason}

Let us take ${\mathcal C}({\mathcal H})$ from Theorem \ref{th:repr},
i.e., the set of closed subspaces of a Hilbert space ${\mathcal H}$.
It is ortho-isomorphic to a Hilbert lattice $\mathcal{HL}$ and
a state in the Hilbert lattice, given by Definition \ref{def:state},
is connected to a state in the Hilbert space as follows:
\begin{eqnarray}
m_\psi(M)=\langle\psi|\hat P_M|\psi\rangle, \quad M\in {\mathcal C}({\mathcal H}),
\label{eq:gleason}
\end{eqnarray}
where $\hat P_M$ denotes the orthoprojector on $\mathcal H$ onto a
closed subspace $M$ that corresponds to a measurable observable,
$|\psi\rangle$ is a unit vector in $\mathcal H$, and
$\langle\psi|\hat P_M|\psi\rangle$ the inner product in $\mathcal H$.
In the quantum physics and quantum computing terminology,
$|\psi\rangle$ is called a state and $\langle\psi|\hat P_M|\psi\rangle$
the amplitude of the probability that the outcome of a measurement
of $\hat P_M$ is in a corresponding Borel set (subset of real numbers),
but we will keep to the Hilbert lattice terminology, in which the
Gleason theorem reads:\ \cite{dvurecenskij-book}

\begin{theorem}\label{th:gleason} {\bf Gleason's Theorem} For any
state $m$ on a Hilbert lattice $\mathcal HL$ of a Hilbert space $\mathcal H$,
${\rm dim}{\mathcal H}\ge 3$, there exists an orthonormal system of
vectors $\{\psi_i\}$ and a system of positive numbers $\{\lambda_i\}$
such that $\sum_i\lambda_i=1$, and
\begin{eqnarray}
m(M)=\sum_i\lambda_im_{\psi_i}(M), \quad M\in {\mathcal HL}.
\label{eq:gleason-t}
\end{eqnarray}
\end{theorem}

Now, subspace $M$ and projector $\hat P_M$ to $M$ correspond
to a measurable observable $\mathcal O$ and a Borel set $E$ whose
values we obtain by a measurement. In other words, $M$ is
determined by $\mathcal O$ and $E$, and we can write $P_M=P^{\mathcal O}_E$.
From $P^{{\mathcal O}_1}_E=P^{{\mathcal O}_2}_E$ it follows
${\mathcal O}_1={\mathcal O}_1$, and from
$P^{\mathcal O}_{E_1}=P^{\mathcal O}_{E_2}$
it follows $E_1=E_1$, so that subspaces from ${\mathcal C}({\mathcal H})$
directly correspond to equivalence classes $|{\mathcal O},E|$ of
$({\mathcal O},E)$ couples. These equivalence classes are elements of
the Hilbert lattice $\mathcal HL$ which is isomorphic to
${\mathcal C}({\mathcal H})$.\ \cite{maczin,ptak-pulm}

However, the axiomatic definition of $\mathcal HL$
by means of universal and existential quantifiers and infinite
dimensionality does not allow us to feed it to a quantum computer.
There\-fore, an attempt has been made to develop an equational
formulation of the Hilbert lattice. The idea is to have infinite
classes of lattice equations that we could use instead.
In applications, infinite classes could then be ``truncated''
to provide us with finite classes of required length. The obtained
classes would in turn contribute to the theory of Hilbert space
subspaces, which so far is poorly developed.

We have considered three ways of reconstructing Hilbert space
starting with an ortholattice. One is geometrical, and
it is presented in Section \ref{sec:goe}. The other is a
probabilistic one (of states, probability measures), and it is
presented in Sections \ref{sec:ge} and \ref{sec:mge}.
They both result in lattice equations that hold in any Hilbert
lattice. The third way is generated by means of vectors
one can define in the Hilbert space isomorphic to the
Hilbert lattice. It includes the fields (real, complex and
the skew field of quaternions) over which the Hilbert space
containing infinite orthonormal sequence of vectors
can only---according to Sol{\`e}r's Theorem \ref{th:sol}---be
defined. This way results in lattice equations that do not
hold in any Hilbert lattice but only those ones that
correspond to a complete description of quantum systems,
those that are difined over the relevant field.

There are four classes of such equations known so far:
the generalized ortho\-arguesian class, the Godowski class,
the Mayet-Godowski class and Mayet's E-class.  Gen\-er\-alized
ortho\-arguesian lattice equations are $n$-variable
equations obtained through extension of 4- and 6- variable
ortho\-argu\-esian equations determined by the pro\-jec\-tive
geometry defined on an ortholattice as presented in
Section \ref{sec:goe}.  Godowski equations and Mayet-Godowski
equations are determined by the states (probability measures)
defined on an ortholattice. They are presented in Sections
\ref{sec:ge} and \ref{sec:mge}. Mayet's E-equations are
determined by a mapping from elements of a Hilbert lattice
to vectors of a Hilbert space defined over one of three
possible fields. They are presented in Section \ref{sec:mayet}.

The Godowski class of equations is included in the Mayet-Godowski
class, but unlike the elegant formalization of the former, so far no
simple characterization of the equational basis for the latter has been
found.  We have shown a simple recursive way to generate all
Mayet-Godowski equations (MGEs and their condensed state equation
representations), but identifying from among them those that are
independent still requires an extensive computational search.
Mayet's specific E-equations presented above can also be given an
elegant recursive formalization. However, as shown by Mayet
\cite{mayet06}, there are other E-equations whose generation
principles are still an open problem.

On the other hand, the techique of obtaining our
equations also suggests a possible third way of generating new
classes of equations. We have seen in Sections \ref{sec:ge} and
\ref{sec:mge} that Greechie diagrams of orthomodular lattices
help us to arrive at new lattice equations. Such finite lattices
have one advantage over the infinite lattices involved in the
definition of Hilbert lattices in Section \ref{sec:hl}.
They enable verification of expressions containing quantifiers,
and we have written programs that can verify such conditions of
Definition  \ref{def:hl}, e.g., superpositions (a) and (b), on any
lattice. This technique could eventually take us to an
exhaustive generation of all classes of lattice equations
that hold in a Hilbert lattice, i.e., to a lattice equation
definition of the Hilbert lattice.  While this open problem
may eventually prove impossible, it still may be possible to
replace some of the quantified conditions with weaker ones,
making up the difference with new lattice equations.
An example of how a quantified condition may be expressed
with an equation is provided by the OML law, which can be
equivalently stated as: \cite[p.~132]{maeda}
\begin{eqnarray}
a \le b  & \Rightarrow &  (\exists c) (a \le c'\ \&\ b = a \cup c).
\end{eqnarray}

Open problems that emerge from the presented research are:
\begin{itemize}
\item{Find any other infinite class of equations, especially the
one that would correspond to the Superposition principle of
Definition \ref{def:hl}.}
\item{Find all classes of equations that hold in a Hilbert
lattice and prove that they are equivalent to the Hilbert lattice
itself.}
\item{Find a geometric interpretation of $n$OA. (A geometric
interpretation of 4OA and 3OA can be inferred from the Arguesian
law, but $n$-dimensional Arguesian law apparently has not been
given an interpretation in the lit\-er\-at\-ure.)}
\item{Find a ``simple characterization'' of finite OMLs that
violate the $n$OA laws (analogous to the wagon-wheel series
for $n$-Go or the lack of a strong set of states for MGEs).}
\item{Prove Theorem~\ref{th:stronger} for any $n$.}
\item{Prove the orthoarguesian identity conjecture (see the
discussion following Theorem~\ref{th:oa-equiv}).}
\item{Prove the conjecture mentioned below Theorem~\ref{th:oa-not-3go}.}
\item{Find a correspondence between Hilbert lattice conditions
and qubit states.}
\item{Determine the complete description of lattices for simple
quantum systems.\footnote{Hultgren and Shimony gave a partial
structure of a lattice for spin-1 system corresponding to a
Stern-Gerlach measurement by means of a magnetic
measurement.\ \cite{shimony} Swift and Wright then showed that
we have to apply both magnetic and electric field to a spin-1
system in a generalized Stern-Gerlach measurement, if we want
to get and measure all lattice elements (propositions).\
\cite{anti-shimony} However, they did not completely describe
the lattice given by Hultgren and Shimony either. See also the
harmonic oscillator example Given by Holland.\ \cite{holl70}}}
\item{Find an equivalent to $n$GO dynamic programing (see
Section \ref{sec:mge}) for $n$OA.}
\item{Find out whether the lattice equations of the $n$th
order can inherently speed up the computation of, say
molecular states, assuming that they would simulate the states
up to a desired precision depending on a chosen $n$.}
\item{Determine if the set of all equations related to
strong sets of $\mathcal{RH}$ states can be given
a simple, universal structure analogous to the
condensed state equations that describe all MGEs.}
\item{Determine whether the new
equations Mayet obtained by direct-sum decompositions \cite{mayet06}
are independent from the $n$OA laws (Theorem~\ref{th:noa}).}
\item{Find specific families of MGEs based on lattice patterns,
analogous to the family $E_n$ and the sequence of
finite OMLs it is based on.}
\item{Answer several open questions asked by Mayet
regarding the independence of some of his equations
related to strong sets of $\mathcal{RH}$ states.\ \cite{mayet06}}
\end{itemize}


\begin{thebibliography}{10}

\bibitem{abrams-lloyd-99}
Daniel~S. Abrams and Seth Lloyd.
\newblock Quantum algorithm providing exponential speed increase for finding
  eigenvalues and eigenvectors.
\newblock {\em {\it Phys. Rev. Lett.}}, {\bf 83}:5162--5165, 1999.

\bibitem{beltr-cass-book}
Enrico~G. Beltrametti and Gianni Cassinelli.
\newblock {\em The Logic of Quantum Mechanics}.
\newblock Addison-Wesley, 1981.

\bibitem{beran}
Ladislav Beran.
\newblock {\em Orthomodular Lattices; {A}lgebraic Approach}.
\newblock D. Reidel, Dordrecht, 1985.

\bibitem{berman-ssqc-book-05-ql}
Gennady~P. Berman, Dmitry~I. Kamenev, and Vladimir~I. Tsifrinovich.
\newblock {\em Perturbation Theory for Solid-State Quantum Computation with
  Many Quantum Bits}.
\newblock Rinton-Press, Paramus, NJ, 2005.
\newblock See Sections 3.1.1 and 4.1.3.

\bibitem{schr-simul}
Bruce~B. Boghosian and Washington {Taylor IV}.
\newblock Simulating quantum mechanics on a quantum computer.
\newblock {\em {\it Physica D}}, {\bf 120}:30--42, 1998.

\bibitem{bouw-ek-zeil-book-00}
Dik Bouwmeester, Artur Ekert, and Anton Zeilinger, editors.
\newblock {\em The Physics of Quantum Information}.
\newblock Springer, Berlin, 2000.

\bibitem{browne05}
Daniel~E. Browne and Terry Rudolph.
\newblock Resource-efficient linear optical quantum computation.
\newblock {\em {\it Phys. Rev. Lett.}}, {\bf 95}:010501--1--4, 2005.

\bibitem{chiara-qcomp04}
Gianpiero Cattaneo, Maria~Luisa {Dalla Chiara}, Roberto Giuntini, and Roberto
  Leporini.
\newblock An unsharp logic from quantum computation.
\newblock {\em {\it Int. J. Theor. Phys.}}, {\bf 43}:1803--1817, 2004.

\bibitem{dvurecenskij-book}
Anatolij Dvure{\v c}enskij.
\newblock {\em Gleason's Theorem and Its Applications}.
\newblock Kluwer, Dordrecht, 1993.

\bibitem{fio-prl04}
Marco Fiorentino and Franco N.~C. Wong.
\newblock Deterministic controlled-{NOT} gate for single-photon two-qubit
  quantum logic.
\newblock {\em {\it Phys. Rev. Lett}}, {\bf 93}:070502--1--4, 2004.

\bibitem{franson-02}
James~D. Franson, M.~M. Donegan, M.~J. Fitch, B.~C. Jacobs, and T.~B. Pittman.
\newblock High-fidelity quantum logic operations using linear optical elements.
\newblock {\em {\it Phys. Rev. Lett.}}, {\bf 89}:137901--1--4, 2002.

\bibitem{gleason}
A.~M. Gleason.
\newblock Measures on the closed subspaces of a {H}ilbert space.
\newblock {\em {\it J. Math. Mechanics}}, {\bf 6}:885--893, 1957.

\bibitem{godow}
Radoslaw Godowski.
\newblock Varieties of orthomodular lattices with a strongly full set of
  states.
\newblock {\em {\it Demonstratio Math.}}, {\bf 14}:725--733, 1981.

\bibitem{go-gr}
Radoslaw Godowski and Richard Greechie.
\newblock Some equations related to the states on orthomodular lattices.
\newblock {\em {\it Demonstratio Math.}}, {\bf 17}:241--250, 1984.

\bibitem{gramss}
Tino Gramss.
\newblock Solving the {S}chr{\"o}dinger equation for the {F}eynman quantum
  computer.
\newblock {\em {\it Int. J. Theor. Phys.}}, {\bf 37}:1423--1439, 1998.

\bibitem{gr-non-s}
Richard~J. Greechie.
\newblock A non-standard quantum logic with a strong set of states.
\newblock In Enrico~G. Beltrametti and B.~C. van Fraassen, editors, {\em
  Current issues in quantum logic, ({P}roceedings of the Workshop on Quantum
  Logic held in {E}rice, {S}icily, {D}ecember 2--9, 1979, at {E}ttore
  {M}ajorana {C}entre for {S}cientific {C}ulture)}, pages 375--380. Plenum
  Press, New York, 1981.

\bibitem{gross}
Herbert Gross.
\newblock Hilbert lattices: {N}ew results and unsolved problems.
\newblock {\em {\it Found. Phys.}}, {\bf 20}:529--559, 1990.

\bibitem{gudder-03}
Stanley Gudder.
\newblock Quantum computational logic.
\newblock {\em {\it Int. J. Theor. Phys.}}, {\bf 42}:39--47, 2003.

\bibitem{holl70}
Samuel~S. {Holland, JR.}
\newblock The current interest in orthomodular lattices.
\newblock In J.~C. Abbot, editor, {\em Trends in Lattice Theory}, pages
  41--126. Van Nostrand Reinhold, New York, 1970.

\bibitem{holl95}
Samuel~S. {Holland, JR.}
\newblock Orthomodularity in infinite dimensions; a theorem of {M}.
  {S}ol{\`e}r.
\newblock {\em {\it Bull. Am. Math. Soc.}}, {\bf 32}:205--234, 1995.

\bibitem{shimony}
Bror~O. {Hultgren, III} and Abner Shimony.
\newblock The lattice of verifiable propositions of the spin-1 system.
\newblock {\em {\it J. Math. Phys.}}, {\bf 18}:381--394, 1977.

\bibitem{ivertsj}
P.-A. Ivert and T.~Sj{\"o}din.
\newblock On the impossibility of a finite propositional lattice for quantum
  mechanics.
\newblock {\em {\it Helv. Phys. Acta}}, {\bf 51}:635--636, 1978.

\bibitem{jauch}
Josef~M. Jauch.
\newblock {\em Foundations of Quantum Mechanics}.
\newblock Addison-Wesley, Reading, Massachusetts, 1968.

\bibitem{kalmb83}
Gudrun Kalmbach.
\newblock {\em Orthomodular Lattices}.
\newblock Academic Press, London, 1983.

\bibitem{kalmb86}
Gudrun Kalmbach.
\newblock {\em Measures and Hilbert Lattices}.
\newblock World Scientific, Singapore, 1986.

\bibitem{kalmb98}
Gudrun Kalmbach.
\newblock {\em Quantum Measures and Spaces}.
\newblock Kluwer, Dordrecht, 1998.

\bibitem{keller}
H.~A. Keller.
\newblock Ein {n}icht-{k}lassischer {H}ilbertscher {R}aum.
\newblock {\em {\it Math. Z.}}, {\bf 172}:41--49, 1980.

\bibitem{ludwig-book-1}
G{\"u}nther Ludwig.
\newblock {\em An axiomatic basis for quantum mechanics. Vol. 1, Derivation of
  Hilbert space structure}.
\newblock Springer, New York, 1985.

\bibitem{ludwig-book-2}
G{\"u}nther Ludwig.
\newblock {\em An axiomatic basis for quantum mechanics. Vol. 2, Quantum
  mechanics and macrosystems}.
\newblock Springer, New York, 1987.

\bibitem{mackey}
George~Whitelaw Mackey.
\newblock {\em The Mathematical Foundations of Quantum Mechanics}.
\newblock W. A. Benjamin, New York, 1963.

\bibitem{maclaren}
M.~Donald Mac{L}aren.
\newblock Atomic orthocomplemented lattices.
\newblock {\em {\it Pacif. J. Math.}}, {\bf 14}:597--612, 1964.

\bibitem{maczin}
Maciej~J. M{\c a}czy{\'n}ski.
\newblock Hilbert space formalism of quantum mechanics without the {H}ilbert
  space axiom.
\newblock {\em {\it Rep. Math. Phys.}}, {\bf 3}:209--219, 1972.

\bibitem{maeda}
Fumitomo Maeda and Sh{\^u}ichir{\^o} Maeda.
\newblock {\em Theory of Symmetric Lattices}.
\newblock Springer-{V}erlag, New York, 1970.

\bibitem{mayet85}
Ren{\'e} Mayet.
\newblock Varieties of orthomodular lattices related to states.
\newblock {\em {\it Algebra Universalis}}, {\bf 20}:368--396, 1985.

\bibitem{mayet86}
Ren{\'e} Mayet.
\newblock Equational bases for some varieties of orthomodular lattices relatied
  to states.
\newblock {\em {\it Algebra Universalis}}, {\bf 23}:167--195, 1986.

\bibitem{mayet98}
Ren{\'e} Mayet.
\newblock Some characterizations of underlying division ring of a {H}ilbert
  lattice of authomorphisms.
\newblock {\em {\it Int. J. Theor. Phys.}}, {\bf 37}:109--114, 1998.

\bibitem{mayet06}
Ren{\'e} Mayet.
\newblock Equations holding in {H}ilbert lattices.
\newblock {\em {\it Int. J. Theor. Phys.}}, {\bf 45}:1216--1246, 2006.

\bibitem{mayet06-hql2}
Ren{\'e} Mayet.
\newblock Ortholattice equations and {H}ilbert lattices.
\newblock In Kurt Engesser, Dov Gabbay, and Daniel Lehmann, editors, {\em
  Handbook of Quantum Logic and Quantum Structures}, volume~{\it Quantum
  Structures}, pages 525--554. Elsevier, Amsterdam, 2007.

\bibitem{bdm-ndm-mp-1}
Brendan~D. Mc{K}ay, Norman~D. Megill, and Mladen Pavi{\v c}i{\'c}.
\newblock Algorithms for {G}reechie diagrams.
\newblock {\em {\it Int. J. Theor. Phys.}}, {\bf 39}:2381--2406, 2000.

\bibitem{mpoa99}
Norman~D. Megill and Mladen Pavi{\v c}i{\'c}.
\newblock Equations, states, and lattices of infinite-dimensional {H}ilbert
  space.
\newblock {\em {\it Int. J. Theor. Phys.}}, {\bf 39}:2337--2379, 2000.

\bibitem{mp-gen-godowski06-arXiv}
Norman~D. Megill and Mladen Pavi{\v c}i{\'c}.
\newblock {M}ayet-{G}odowski {H}ilbert lattice equations.
\newblock {\em {\tt arXiv.org/quant-ph/0609192}}, 2006.

\bibitem{nielsen-chuang-ql482}
Michael~A. Nielsen and Isaac~L. Chuang.
\newblock {\em Quantum Computation and Quantum Information}.
\newblock Cambridge University Press, Cambridge, 2000.
\newblock See p.~482.

\bibitem{pavicic-book-05}
Mladen Pavi{\v c}i{\'c}.
\newblock {\em Quantum Computation and Quantum Communication: {T}heory and
  Experiments}.
\newblock Springer, New York, 2005.

\bibitem{mpcommp99}
Mladen Pavi{\v c}i{\'c} and Norman~D. Megill.
\newblock Non-orthomodular models for both standard quantum logic and standard
  classical logic: Repercussions for quantum computers.
\newblock {\em {\it Helv. Phys. Acta}}, {\bf 72}:189--210, 1999.

\bibitem{pm-ql-l-hql1}
Mladen Pavi{\v c}i{\'c} and Norman~D. Megill.
\newblock {{\it Is Quantum Logic a Logic?}}
\newblock In Kurt Engesser, Dov Gabbay, and Daniel Lehmann, editors, {\em
  Handbook of Quantum Logic and Quantum Structures}, volume~{\it Quantum
  Logic}, pages 23--47. Elsevier, Amsterdam, 2008.

\bibitem{piron-book}
Constantin Piron.
\newblock {\em Foundations of Quantum Physics}.
\newblock Benjamin, Reading, Mass., 1976.

\bibitem{pittenger-book}
Arthur~O. Pittenger.
\newblock {\em An Introduction to Quantum Computing Algorithms}, volume~19 of
  {\em Progress in Computer Science and Applied Logic}.
\newblock Birkh\"auser, Boston, 1999.

\bibitem{ptak-pulm}
Pavel Pt{\'a}k and Sylvia Pulmannov{\'a}.
\newblock {\em Orthomodular Structures as Quantum Logics}.
\newblock Kluwer, Dordrecht, 1991.

\bibitem{ralph05}
T.~C. Ralph, A.~J.~F. Hayes, and Alexei Gilchrist.
\newblock Loss-tolerant optical qubits.
\newblock {\em {\it Phys. Rev. Lett.}}, {\bf 95}:100501--1--4, 1905.

\bibitem{shapiro-prl}
E.~A. Shapiro, Michael Spanner, and Misha~Yu. Ivanov.
\newblock Quantum logic approach to wave packet control.
\newblock {\em {\it Phys. Rev. Lett.}}, {\bf 91}:237901--1--4, 2003.

\bibitem{shor}
P.~W. Shor.
\newblock Polynomial-time algorithms for prime factorization and discrete
  logarithms on a quantum computer.
\newblock {\em {\it {SIAM}} J. Comp.}, {\bf 26}:1484--1509, 1997.

\bibitem{soler}
Maria~Pia Sol{\`e}r.
\newblock Characterization of {H}ilbert spaces by orthomodular spaces.
\newblock {\em {\it Comm. Alg.}}, {\bf 23}:219--243, 1995.

\bibitem{summhammer}
Johann Summhammer.
\newblock Factoring and {F}ourier transformation with a {M}ach-{Z}ehnder
  interferometer.
\newblock {\em {\it Phys. Rev. A}}, {\bf 56}:4324--4326, 1997.

\bibitem{svozil-tkadlec}
Karl Svozil and Josef Tkadlec.
\newblock Greechie diagrams, nonexistence of measures and
  {K}ochen--{S}pecker-type constructions.
\newblock {\em {\it J. Math. Phys.}}, {\bf 37}:5380--5401, 1996.

\bibitem{anti-shimony}
Arthur~R. Swift and Ron Wright.
\newblock Generalized {S}tern-{G}erlach experiments and the observability of
  arbitrary spin operators.
\newblock {\em {\it J. Math. Phys.}}, {\bf 21}:77--82, 1980.

\bibitem{varad}
V.~S. Varadarajan.
\newblock {\em Geometry of Quantum Theory, Vols. 1 \&\ 2}.
\newblock John Wiley \& Sons, New-York, 1968,1970.

\bibitem{zalka-98}
Christof Zalka.
\newblock Simulating quantum systems on a quantum computer.
\newblock {\em {\it Proc. Roy. Soc. London A}}, {\bf 454}:313--322, 1998.

\bibitem{zurek-96}
Wojciech~Hubert Zurek and Raymond Laflamme.
\newblock Quantum logical operations on encoded qubits.
\newblock {\em {\it Phys. Rev. Lett.}}, {\bf 77}:4683--4686, 1996.

\end{thebibliography}
\end{document}